\documentclass[%
  reprint,
  amsmath,amssymb,
  aps,
]{revtex4-2}

\usepackage{graphicx}
\usepackage{dcolumn}
\usepackage{bm}

\usepackage{blindtext}
\usepackage{amsfonts,amsthm,mathtools}
\usepackage{physics}
\usepackage{float}
\usepackage{comment}
\usepackage{enumitem}
\usepackage{xcolor}
\usepackage{bbm}
\usepackage{stackengine}
\usepackage{adjustbox}
\usepackage{tabularx,booktabs,siunitx,makecell}
\usepackage{quantikz}
\usepackage{url}
\usepackage{algorithm}
\usepackage{algorithmicx}
\usepackage[noend]{algpseudocode}

\allowdisplaybreaks

\usepackage{hyperref}
\hypersetup{
  colorlinks=true,
  breaklinks=true,
  urlcolor=blue,
  linkcolor=blue,
  citecolor=red,
}

\newtheorem{proposition}{Proposition}
\newtheorem{definition}{Definition}
\newtheorem{remark}{Remark}

\begin{document}

\title{Provably Optimal Quantum Circuits with Mixed-Integer Programming}

\author{Harsha Nagarajan}
\email[Corresponding author: ]{harsha@lanl.gov}
\affiliation{Applied Mathematics and Plasma Physics (T-5), Los Alamos National Laboratory, Los Alamos, New Mexico, USA}

\author{Zsolt Szab\'o}
\affiliation{ARC Centre of Excellence for Engineered Quantum Systems, Macquarie University, Sydney, NSW 2109, Australia}
\affiliation{School of Mathematical and Physical Sciences, Macquarie University, NSW 2109, Australia}


\begin{abstract}
  We present a depth-aware optimization framework for quantum circuit compilation that unifies provable optimality with scalable heuristics. For exact synthesis of a target unitary, we formulate a mixed-integer linear program (MILP) that \emph{linearly} handles global-phase equivalence and uses explicit parallel scheduling variables to certify depth-optimal solutions for small–to–medium circuits. Domain-specific valid inequalities—including identity ordering, commuting-gate pruning, short-sequence redundancy cuts, and Hermitian-conjugate linkages—significantly tighten the relaxation and accelerate branch-and-bound, yielding speedups up to $43\times$ on standard benchmarks. The framework supports hardware-aware objectives, enabling fault-tolerant priorities (e.g., $T$-count) and NISQ-era penalties (e.g., entangling gates). For approximate synthesis, we propose three fidelity-driven objectives: (i) exact—but non-convex—phase-invariant fidelity maximization; (ii) a linear surrogate that maximizes the real trace overlap, yielding a tight lower bound to fidelity; and (iii) a convex quadratic function that minimizes the circuit's Frobenius error. 
  
  To scale beyond exact MILP, we propose a novel rolling-horizon optimization (RHO) that rolls primarily in time, caps the active-qubit set, and enforces per-qubit closure while globally optimizing windowed segments. This preserves local context, reduces the effective Hilbert-space dimension, and enables iterative improvements without ancillas. On a 142-gate seed circuit, RHO yields 116 gates—an 18.3\% reduction from the seed—while avoiding the trade-off between myopic passes and prohibitive solve times. Empirically, our exact compilation framework achieves certified depth-optimal decompositions on standard targets, high-fidelity Fibonacci-anyon weaves, and a 36\% gate-count reduction on multi-body parity circuits. All methods are implemented in the open-source package \texttt{QuantumCircuitOpt}, providing a single optimization framework that bridges exact certification and hardware-aware, scalable synthesis.
  \end{abstract}


\maketitle

\section{Introduction}\label{sec:intro}

Quantum computation holds the potential to revolutionize information processing and the solution of complex problems. A major obstacle to its practical realization, however, is the compilation of quantum algorithms into executable low-level instructions for physical hardware. This task—known as \emph{quantum compilation}—involves translating high-level descriptions of quantum algorithms into ordered sequences of native gates compatible with a given device, a process complicated by the constraints of current quantum hardware (restricted qubit connectivity, gate infidelities, decoherence) and the fragility of quantum states that are used to store and process information. Within this context, the research community has focused heavily on \(T\) gates, motivated by fault-tolerant quantum computing practices~\cite{jones_layered_2012, campbell_roads_2017}. In surface-code architectures, non-Clifford operations such as the \(T\) gate are typically implemented via magic-state distillation, which dominates both spatial and temporal resource costs. Consequently, minimizing \(T\)-count—and its parallelized counterpart, \(T\)-depth—remains a central objective in quantum circuit optimization. In contrast, on NISQ-era devices without error correction, two-qubit gates (e.g., CNOT) are typically the dominant contributors to error and latency; minimizing their count and depth is therefore critical for performance. Given these hardware-specific cost models, efficient circuit compilation is essential to realizing a practical quantum advantage.

Heuristic approaches have demonstrated strong practical efficacy in reducing quantum circuit size and depth across hardware platforms. QFAST combines combinatorial search with numerical optimization to find low-depth implementations for moderate-scale circuits at high fidelity~\cite{younis_qfast_2021}. Hardware-accelerated data-flow engines on Field Programmable Gate Arrays (FPGAs) extend variational synthesis to 3–9 qubits while maintaining near-unit fidelities via parallel optimization~\cite{rakyta_highly_2024}. Graph-theoretic methods based on the ZX-calculus systematically simplify circuits by merging phase gadgets and optimizing non-Clifford structure~\cite{duncan_graph-theoretic_2020}. More recently, AlphaTensor-Quantum, a deep reinforcement learning framework, has reformulated circuit identity discovery as tensor optimization problems, demonstrating empirical improvements in $T$-gate count reduction through learned policies~\cite{ruiz_quantum_2025}.
Despite their scalability—especially ZX rewriting and AlphaTensor-Quantum—and non-trivial $T$-gate savings, these methods provide neither optimality guarantees nor a posteriori certificates on solution quality.

Exact synthesis algorithms provide optimality certificates only for constrained settings (small instances or restricted gate sets). The meet-in-the-middle search yields depth-optimal Clifford+$T$ decompositions for small multi-qubit circuits~\cite{amy_meet---middle_2013}, while number-theoretic synthesis achieves near-optimal single-qubit $Z$-rotation approximations in Clifford+$T$ via Diophantine equation solutions~\cite{ross_optimal_2016}. Polynomial-time optimizers also exist for $T$-depth minimization in specific circuit structures \cite{amy_polynomial-time_2014}. However, the general problem is provably intractable: $T$-count and $T$-depth minimization for exact unitary synthesis are both NP-hard problems~\cite{amy_polynomial-time_2014}, with the hardness stemming from the combinatorial explosion in valid gate sequences subject to quantum mechanical constraints. This computational barrier explains the field's reliance on heuristic approaches for circuits beyond a trivial scale.

Decades of mathematical advances in mixed-integer programming (MIP) have established it as a natural framework for high-dimensional, combinatorial design problems like quantum circuit compilation. Its applicability extends beyond circuit synthesis—qubit placement and routing, for instance, can be formulated as integer linear programs that jointly optimize initial layout and SWAP insertion~\cite{nannicini_optimal_2023}. MIP modeling provides four key advantages for quantum compilation: (i) Unified modeling of discrete choices and continuous algebraic constraints, (ii) Mature solvers with advanced presolve, cutting planes (or cuts), and parallel branch-and-bound, (iii) Flexible objectives beyond $T$-count (e.g., depth minimization, entangling-gate penalties, hardware-aware weights, fidelity), and (iv) Provable optimality through solution certificates or bounded optimality gaps for tractable instances.  

Nagarajan \textit{et al.} introduced a MIP formalism for circuit optimization in the open-source package \texttt{QuantumCircuitOpt} (\texttt{QCOpt}), providing a rigorous approach to circuit design with optimality guarantees~\cite{nagarajan_quantumcircuitopt_2021}. Subsequently, \cite{henderson_exploring_2022} explored continuous/nonlinear formulations and relaxations that accelerate search while preserving solution quality. Together, these works establish a baseline for exact certification and an extensible modeling stack that we build upon here.

We contribute to \texttt{QCOpt} in five directions. 
(1) We formulate depth as a principal objective via explicit scheduling/precedence variables and per-depth qubit–disjointness, certifying depth-optimal solutions at modest scales. 
(2) We handle global-phase invariance \emph{linearly} in a real embedding, eliminating non-convex phase constraints. 
(3) We propose a catalog of novel valid inequalities—identity ordering, commuting/equivalent-pattern pruning, short-sequence redundancy cuts, and Hermitian-conjugate linkages—that significantly tighten the MILP and accelerate branch-and-bound-based solvers. 
(4) For approximate synthesis, we formulate fidelity-driven objectives: an exact (non-convex) phase-invariant fidelity in real encoding, a linear real-part surrogate, and a piecewise-linear outer approximation of Frobenius error; we also show that phase-optimized Frobenius minimization is equivalent to fidelity's real-part maximization, and 
(5) For scale, we introduce a rolling-horizon algorithm that rolls primarily in time, caps the active-qubit set, and enforces per-qubit closure, preserving the tightened constraints and enabling iterative improvements on larger circuits.

Our framework is inherently adaptable to diverse hardware cost models through objective and gate-weight selection, accommodating both NISQ-era regimes—where two-qubit gate errors dominate—and fault-tolerant architectures, where non-Clifford resources (notably $T$ gates) are the principal cost metric. The proposed formulation targets fundamental \emph{unitary synthesis} rather than post-hoc circuit optimization: given a target unitary and an elementary gate set, the mixed-integer program selects an optimal gate ordered sequence that implements the target exactly (modulo global phase) or within a specified approximation tolerance. Hardware topology is incorporated natively by restricting multi-qubit gate placement to edges of the device coupling graph, thereby enforcing connectivity by construction rather than via post-processing.

These capabilities—together with our depth-optimization objective and catalog of domain-specific valid inequalities—yield provably optimal solutions for small-to-moderate circuits while preserving strong empirical performance at larger scales. This ensures immediate practical value across both near-term platforms and fault-tolerant architectures, with all details in the sections that follow.

\section{Mathematical Formulation} 
\label{sec:math_form}
\noindent
\textbf{Mathematical Preliminaries}: For any integer $n \geqslant 1$, we denote $[n] \coloneqq \{1, 2, \ldots, n\}$ as the set of the first $n$ positive integers, and thus $[n] \setminus \{1\} = \{2, \dots, n\}$. Let \(\mathbbm{1}_n\) be an \(n\times n\) identity matrix. Throughout, $A^\dagger = \overline{A}^{\mathsf{T}}$ denotes the Hermitian-conjugate; for unitary $U \in \mathbb{C}^{n \times n}$, $U^\dagger U = UU^\dagger = \mathbbm{1}_n$. 

Consider a quantum register of $Q$ qubits with associated Hilbert space $\mathcal{H} = (\mathbb{C}^2)^{\otimes Q} \cong \mathbb{C}^{2^Q}$. For any qubit subset $S \subseteq [Q]$, define the unitary group on $S$ and its special unitary subgroup as
\begin{align}
\mathcal{U}(2^{|S|})
& = \left\{ U \in \mathbb{C}^{2^{|S|} \times 2^{|S|}} \;\middle|\; U^\dagger U = \mathbbm{1}_{2^{|S|}} \right\}, \\ 
\mathcal{SU}(2^{|S|})
& = \left\{ U \in \mathcal{U}(2^{|S|}) \;\middle|\; \det U = 1 \right\}.
\end{align}
Let $\mathbb{G} \subset \bigcup_{S \subseteq [Q]} \mathcal{U}(2^{|S|})$ be a finite elementary gate set with fixed parameters. For each gate $g \in \mathbb{G}$ acting on subset $S_g \subseteq [Q]$, its extension to the full Hilbert space $\mathcal{H}$ is denoted $g^{(S_g)} \in \mathcal{U}(2^Q)$ as defined in Definition \ref{defn:gate_extn}. For each qubit $q \in [Q]$, denote $\mathbb{G}_q = \{g \in \mathbb{G} \mid g \text{ acts non-trivially on } q\}$ as the subset of gates affecting qubit $q$ (see Definition \ref{defn:trivial_action}).

\textit{Quantum circuit compilation} seeks an ordered sequence of gates $(U_p)_{p=1}^P$ where each $U_p$ corresponds to the full Hilbert space representation of some gate from $\mathbb{G}$, such that:
\begin{equation}\label{eq:target}
\prod_{p=1}^{P} U_p = e^{i\phi} \widetilde T \quad \phi \in [0,2\pi),
\end{equation}
where $ \widetilde T \in \mathcal{U}(2^Q)$ is the given target unitary, and $P$ is the maximum allowable gate count. Specifically, for each position $p \in [P]$, there exists $g_p \in \mathbb{G}$ such that $U_p = g_p^{(S_{g_p})}$. The circuit is structured across $P$ positions $p \in [P]$, with depth $D$ of the circuit as defined in Definition~\ref{defn:depth}.

The optimization landscape for quantum compilation encompasses multiple physically significant metrics, each addressing critical constraints of near-term and fault-tolerant quantum hardware:

\begin{itemize}[leftmargin=*]
\item \textit{Gate count}: Minimize $P$, the total number of elementary operations, to shorten the gate sequence and reduce computational time and error rates.
\item \textit{Entangling gate count}: Minimize entangling operations (e.g., CNOT, CZ), which typically have error rates 2–10× higher than single-qubit gates across platforms, thereby mitigating crosstalk and other interaction-induced errors.
\item \textit{Non-Clifford resources}: In fault-tolerant architectures, minimize $T$ gates (and other non-Clifford gates), each requiring costly magic-state distillation; this is particularly important in error-corrected computation due to their higher resource demands.
\item \textit{Circuit depth}: Minimize the circuit depth $D$ (overall execution time) by maximal parallelism, thereby mitigating decoherence and dephasing in the algorithm.
\item \textit{Circuit fidelity}: Under fixed resource and hardware constraints, maximize circuit fidelity with the target unitary $\widetilde T$ to ensure optimal performance.
\end{itemize}

To address any of these optimization metrics, the combinatorial complexity of quantum gate selection while respecting hardware constraints and the continuous nature of unitary evolution within a unified optimization framework, we case circuit synthesis as a Mixed-Integer Linear Program (MILP). The complete formalism is  presented in the following subsection.

\subsection{Mixed-Integer Programming Formalism}
\label{subsec:milp}
A Mixed-Integer Linear Program (MILP) optimizes a linear objective subject to affine constraints over mixed continuous-binary variables. The canonical form is:
\begin{equation}\label{eq:milp-standard}
\begin{aligned}
\min_{\mathbf{x},\mathbf{z}} \quad & \mathbf{c}^\top\mathbf{z} \\
\text{s.t.} \quad & A\mathbf{x} + B\mathbf{z} = \mathbf{b} \\
& \mathbf{x}^\mathrm{L} \leqslant \mathbf{x} \leqslant \mathbf{x}^\mathrm{U} \\
& \mathbf{z} \in \{0,1\}^{n_z}
\end{aligned}
\end{equation}
where $A \in \mathbb{R}^{m \times n_x}$, $B \in \mathbb{R}^{m \times n_z}$, $\mathbf{b} \in \mathbb{R}^m$, $\mathbf{c} \in \mathbb{R}^{n_z}$, and $\mathbf{x}^\mathrm{L}, \mathbf{x}^\mathrm{U} \in (\mathbb{R} \cup \{-\infty, +\infty\})^{n_x}$ constitute the problem data. Continuous variables $\mathbf{x} \in \mathbb{R}^{n_x}$ represent real-encoded unitary entries, while binary variables $\mathbf{z} \in \{0,1\}^{n_z}$ encode gate selection and depth assignments. The constraint matrix $[A\ B]$ enforces quantum-mechanical validity through linear relationships derived in detail in Section \ref{subsec:mip_constraints}. When constraints or the objective function contain quadratic (bilinear) terms, the formulation becomes a Mixed-Integer Quadratic Program (MIQP).

Although solving MILPs to optimality is NP-hard in the worst case, modern commercial solvers (Gurobi \cite{Gurobi}, CPLEX \cite{Cplex}) routinely handle large, structured instances via branch-and-cut: they explore a branch-and-bound tree while repeatedly solving tightened linear-programming (continuous) relaxations and adding cutting planes. Decades of advances—presolve, cut separation, primal heuristics, and sophisticated branching—have yielded performance improvements outpacing Moore’s-law hardware gains \cite{koch2022progress}. Our contribution is a tight MILP formulation for quantum circuit synthesis that exactly captures the compilation problem's combinatorial and continuous structure. While this exact representation theoretically guarantees optimal circuits, its $\Theta(4^Q)$ scaling with qubit count $Q$ limits practical applicability. We therefore couple the formulation with quantum-specific enhancements that preserve optimality while enabling synthesis of larger circuits, demonstrating that rigorous MILP modeling combined with domain-aware acceleration yields a scalable exact synthesis framework for non-trivial quantum circuits.

\subsection{Real Encoding of Complex Matrices}
\label{subsec:real-encoding}

To represent complex matrices using real arithmetic while preserving the algebraic structure of matrix operations, we employ a real encoding where each complex entry is mapped to a $2\times 2$ real block. Define the complex-to-real map
\begin{equation}
\label{eq:calR-scalar}
\mathcal{R}:\ \mathbb{C}\to\mathbb{R}^{2\times 2},\qquad 
\mathcal{R}(a+ib)=
\begin{bmatrix}
a & -b\\
b & \ \,a
\end{bmatrix}.
\end{equation}
For \(A\in\mathbb{C}^{n\times n}\) with entries \(A_{ij}\), the elementwise block encoding \(\mathcal{R}(A)\in\mathbb{R}^{2n\times 2n}\) is the \(n\times n\) array of \(2\times 2\) blocks defined by
\begin{equation}\label{eq:calR-elem-embed}
\big(\mathcal{R}(A)\big)_{[2i-1:2i,\;2j-1:2j]} \;=\; \mathcal{R}(A_{ij})
\quad \forall \ i,j\in [n].
\end{equation}
Equivalently, with \(E_{ij}\) the matrix units, \(\mathcal{R}(A)=\sum_{i,j} E_{ij}\otimes \mathcal{R}(A_{ij})\).
Up to a fixed permutation \(\Pi\), this coincides with the canonical form:
\(\mathcal{R}(A)=\Pi\begin{bmatrix}\Re A&-\Im A\\ \Im A&\Re A\end{bmatrix}\Pi^{\mathsf T}\).

\medskip
\noindent\textbf{Basic properties.}
We write \(A^\dagger=\overline{A}^{\mathsf T}\) for the conjugate transpose. The map \(\mathcal{R}\) is an injective \(*\)-algebra homomorphism:
\begin{align}\label{eq:calR-props}
\mathcal{R}(A+B)&=\mathcal{R}(A)+\mathcal{R}(B), &
\mathcal{R}(AB)&=\mathcal{R}(A)\,\mathcal{R}(B),\\
\mathcal{R}(A^\dagger)&=\mathcal{R}(A)^{\mathsf T}, &
\mathcal{R}(\mathbbm{1}_n)&=\mathbbm{1}_{2n}. \nonumber
\end{align}

\medskip
Consequently, if \(A\) is unitary, then
\(
\mathcal{R}(A)^{\mathsf T}\mathcal{R}(A)=\mathcal{R}(A^\dagger A)=\mathcal{R}(\mathbbm{1}_n)=\mathbbm{1}_{2n},
\)
so \(\mathcal{R}(A)\) is orthogonal. Moreover, the following proposition holds true: 
\begin{proposition}
$\det \mathcal{R}(A)=|\det A|^2.$
\label{prop:calR-det}
\end{proposition}
\begin{proof}
By Schur decomposition \(A=Q T Q^\dagger\) with \(Q\) unitary and \(T\) upper triangular. Using multiplicativity, \(\det\mathcal{R}(A)=\det\mathcal{R}(Q)\det\mathcal{R}(T)\det\mathcal{R}(Q^\dagger)\).
Since \(Q\) is unitary, \(\det \mathcal{R}(Q)=\det \mathcal{R}(Q^\dagger)=1\).
For triangular \(T\) with diagonal \((\tau_1,\dots,\tau_n)\), \(\det\mathcal{R}(T)=\prod_{k=1}^n \det\mathcal{R}(\tau_k)=\prod_{k=1}^n |\tau_k|^2=|\det T|^2=|\det A|^2\).
\end{proof}
Hence we have \(A\in \mathcal{SU}(n)\Rightarrow \det \mathcal{R}(A)=1\).
\medskip

For a \textit{global complex phase} $\lambda = e^{i\phi}$ ($|\lambda| = 1$), with $\lambda = r + i s$ for $r,s \in \mathbb{C}$, then
\begin{equation}\label{eq:calR-phase}
\mathcal{R}(\lambda A)\;=\; \big(\mathbbm{1}_n \otimes \mathcal{R}(\lambda)\big)\,\mathcal{R}(A)
\;=\; \mathcal{R}(A)\,\big(\mathbbm{1}_n \otimes \mathcal{R}(\lambda)\big).
\end{equation}
Thus a global-phase \(\lambda\) acts blockwise as the planar rotation
\(
\mathcal{R}(e^{i\phi})=\begin{bmatrix}\cos\phi&-\sin\phi\\ \sin\phi&\cos\phi\end{bmatrix},
\)
so any phase-invariant constraints in the complex model translate directly under this encoding.

\begin{proposition}\label{prop:Trace-identities}
Define
$J_n \;=\; \mathbbm{1}_n\!\otimes\!\begin{psmallmatrix}0&-1\\[1pt]1&0\end{psmallmatrix}
\;=\; \mathcal{R}(i\,\mathbbm{1}_n).
$
Then
$$
\Re\!\big(\operatorname{Tr}A\big)=\tfrac12\,\operatorname{Tr}\!\big(\mathcal{R}(A)\big),
\quad
\Im\!\big(\operatorname{Tr}A\big)=-\tfrac12\,\operatorname{Tr}\!\big(J_n\,\mathcal{R}(A)\big).
$$
\end{proposition}

\begin{proof}
Write $A_{ij}=a_{ij}+i b_{ij}$. By \eqref{eq:calR-scalar},
$\mathcal{R}(A)$ has $(i,i)$-block $\begin{psmallmatrix}a_{ii}&-b_{ii}\\ b_{ii}&a_{ii}\end{psmallmatrix}$,
whose trace is $2a_{ii}$. Summing over $i$ gives
$\operatorname{Tr}(\mathcal{R}(A))=2\sum_i a_{ii}=2\,\Re(\operatorname{Tr}A)$.

Using this result, for the imaginary part, we have
$
J_n\,\mathcal{R}(A)=\mathcal{R}(i\,\mathbbm{1}_n)\,\mathcal{R}(A)=\mathcal{R}(iA),
$
hence $\operatorname{Tr}(J_n\mathcal{R}(A))=\operatorname{Tr}(\mathcal{R}(iA))
=2\,\Re(\operatorname{Tr}(iA))=2\,\Re(i \operatorname{Tr}(A)) = -2\,\Im(\operatorname{Tr}A)$.
\end{proof}

\noindent
\textbf{Inverse map:} For any $B\in\mathbb{R}^{2n\times 2n}$ with block structure
$B_{2i,\,2j}=B_{2i-1,\,2j-1}$ and $B_{2i-1,\,2j}=-B_{2i,\,2j-1},\ \forall \  i,j\in[n]$,
define
\[
\big(\mathcal{R}^{-1}(B)\big)_{ij}=B_{2i-1,\,2j-1}+i\,B_{2i,\,2j-1}.
\]
Then $\mathcal{R}:\mathbb{C}^{n\times n}\!\to\!\mathbb{R}^{2n\times 2n}$ is a bijection onto the set of real matrices with this structure, with
$\mathcal{R}^{-1}(\mathcal{R}(A))=A$ for all $A$ and
$\mathcal{R}(\mathcal{R}^{-1}(B))=B$ for all such $B$.

\subsection{Linearization of Bilinear Terms}
\label{subsec:linearization}
The optimal circuit compilation problem's constraints in Eq.~\eqref{eq:target} contain multilinear terms and are therefore not directly MILP-representable. We obtain a MILP by replacing each binary–continuous product with its McCormick convex-hull linearization. For $z \in \{0,1\}$ and $x \in [-1,1]$, the set  
\begin{equation}
\mathcal{S} = \big\{ (x,z,\beta) \mid \beta = zx,\ z \in \{0,1\},\ x \in [-1,1] \big\}
\end{equation}
has convex hull  
\begin{equation*}
\label{eq:mccormick}
\operatorname{conv}(\mathcal{S}) = \left\{ (x,z,\beta) \;\middle|\;
\begin{aligned}
&\beta \leqslant  z,\quad \beta \geqslant  -z, \\
&\beta \leqslant  x + 1 - z,\quad \beta \geqslant  x - 1 + z, \\
&x \in [-1,1],\ z \in [0,1]
\end{aligned}
\right\}.
\end{equation*}  
Crucially, when $z$ is binary and $x \in [-1,1]$, $\operatorname{conv}(\mathcal{S})$ exactly enforces $\beta = zx$ due to the boundedness of $x$. The exactness of this linearization also extends to all bilinear binary-binary products ($z_1, z_2 \in \{0,1\}$). Multilinear terms are linearized recursively by introducing auxiliary variables for each binary-continuous factor, preserving exactness \cite{nagarajan2016tightening}. The complete MILP reformulation (with auxiliaries) is in Section \ref{subsec:mip_constraints}.

\subsection{Variables and Constraints}
\label{subsec:mip_constraints}
Given a register of $Q$ qubits and a finite set of elementary gates $\mathbb{G}$ with fixed entries, we seek to construct a quantum circuit composed of at most $P$ gates, placed at positions $p \in [P]$, with an overall depth not exceeding a prescribed maximum depth $D$ (the latter is defined in Definition~\ref{defn:depth} and particularly relevant in depth–optimization tasks). 

\begin{definition}[Trivial vs.\ non-trivial action on qubits]
Let $\mathcal{H}=\mathbb{C}^{2^Q}$ and let $U\in\mathcal{U}(2^Q)$. 
For a qubit subset $S\subseteq[Q]$, we say that $U$ \emph{acts trivially on $S$} iff there exists a unitary $W\in\mathcal{U}(2^{Q-|S|})$ such that
\begin{equation}
U \;=\; P_S^\dagger\!\left(\,\mathbbm{1}_{2^{|S|}}\otimes W\,\right) P_S,
\end{equation}
where $P_S\in\mathcal{U}(2^Q)$ permutes tensor factors so that the qubits in $S$ occupy the first $|S|$ positions (the property is independent of the particular choice of $P_S$). 
If no such $W$ exists, then $U$ \emph{acts non-trivially on $S$}.

In particular, for a single qubit $q\in\{1,\dots,Q\}$, $U$ acts trivially on $q$ iff
\begin{equation}
U \;=\; P_{\{q\}}^\dagger\!\left(\,\mathbbm{1}_{2}\otimes W\,\right) P_{\{q\}}
\quad\text{for some }W\in\mathcal{U}(2^{Q-1}).
\end{equation}
\label{defn:trivial_action}
\end{definition}

\begin{definition}[Gate extension to $\mathcal{H}$]
For an elementary gate $U \in \mathcal{U}(2^{|S|})$ acting non-trivially on qubit subset $S \subseteq \{1,\dots,Q\}$, its extension to full Hilbert space $\mathcal{H}$ is the unitary operator $U^{(S)} \in \mathcal{U}(2^Q)$ defined by  
\begin{equation}
U^{(S)} = P_S^\dagger \left( U \otimes \mathbbm{1}_{2^{Q-|S|}} \right) P_S,
\end{equation}
where $P_S \in \mathcal{U}(2^Q)$ is the permutation that moves the qubits in $S$ (in increasing index order) to the first $|S|$ tensor slots and preserves the relative order of the remaining qubits.
In particular, when $|S| = 1$ (i.e., $S = \{q\}$), $U^{(\{q\})}$ reduces to $U \otimes \mathbbm{1}_{2^{Q-1}}$ in the basis where qubit $q$ is positioned first.
\label{defn:gate_extn}
\end{definition}

\begin{definition}[Circuit depth]
Given a quantum circuit \(\mathcal{C}\) with gate sequence \(g_1, \dots, g_P\) acting on \(Q\) qubits, depth \(D(\mathcal{C})\) is the smallest integer \(D\) for which there exists an assignment of gates to depths \(\delta(p) \in [D]\) such that: 
(i) \(\delta(p) \leqslant  \delta(p+1)\) $\forall$ \(p < P\), and  
(ii) For every qubit \(q \in [Q]\) and depth \(d \in [D]\), at most one gate acting non-trivially on \(q\) is assigned to depth \(d\).
\label{defn:depth}
\end{definition}
Definition \ref{defn:depth} captures the minimal number of sequential time steps required to execute a circuit $\mathcal{C}$ on hardware, where gates acting on disjoint qubits may execute in parallel within the same depth. To formalize this definition in the MIP context, we introduce several key variables and constraints: \\

\noindent
\textbf{Constraints: Gate count minimization} \\ 

\noindent
To facilitate an MILP formulation, we introduce (i) binary selection variables $z_{g,\,p} \in \{0,1\}$ indicating that gate $g\in\mathbb{G}$ is chosen at position $p\in [P]$, and (ii) real matrix variables $G_p\in[-1,1]^{2^{Q+1}\times 2^{Q+1}}$ denoting the fixed real block representation $G_p \equiv \mathcal{R}(U_p)$ of the unitary at position $p$ (see Section~\ref{subsec:real-encoding}). Using the canonical extension of gates $g \in \mathbb{G}$ to the full Hilbert space $\mathcal{H}$ (Definition~\ref{defn:gate_extn}), the gate–assignment constraints are
\begin{subequations}\label{eq:gate_select}
\begin{align}
&\sum_{g\in\mathbb{G}} z_{g,\,p} \;=\; 1, 
\quad \forall\, p\in[P],
\label{eq:gate_select:onehot}
\\
&G_p \;=\; \sum_{g\in\mathbb{G}} z_{g,\,p} \cdot \mathcal{R}(g), 
\quad \forall\, p\in[P].
\label{eq:gate_select:link}
\end{align}
\end{subequations}
Constraint~\eqref{eq:gate_select:onehot} ensures exactly one gate occupies each position, and constraint \eqref{eq:gate_select:link} links the binary choice to the (embedded) gate so that $G_p$ equals the selected element of $\mathbb{G}$ acting on $\mathcal{H}$, encoded as a real matrix. 

Additionally, let variables $\hat{G}_p$ represent the cumulative unitary up to position $p$, also with real values in $[-1,1]^{2^{Q+1} \times 2^{Q+1}}$. To maintain the correct circuit-product structure, we impose:
\begin{align} \label{eq:G-hat}
\hat{G}_1 &= G_1, \nonumber \\
\hat{G}_p &=  \hat{G}_{p-1} \cdot G_p, \quad \forall p \in [P]\setminus \{1\}.
\end{align}
The recursive set of constraints in \eqref{eq:G-hat} is inherently bilinear in continuous-binary products, and thus, we use the successive McCormick linearization technique described in Section \ref{subsec:linearization} to handle them efficiently. \\ 

\noindent
\textbf{Constraints: Circuit depth minimization}  
In addition to the variables and constraints introduced above, for the task of circuit–depth minimization (Definition~\ref{defn:depth}) we introduce assignment binaries $b_{p,\,d}\in\{0,1\}$ indicating that the gate at position $p$ is scheduled at depth $d$ (i.e., $b_{p,\,d}=1$ iff the gate at position $p$ is assigned to depth $d$). The depth-related constraints are:
\begin{subequations}
\label{eq:depth_mip}
\begin{align}
&\sum_{d=1}^{D} b_{p,\,d} \;=\; 1,
\quad \forall\, p \in [P], 
\label{eq:depth_mip:assign}
\\
&0 \;\leqslant \; \sum_{d=1}^{D} d\,\big(b_{p,\,d}-b_{p-1,\,d}\big) \;\leqslant \; 1,
\ \forall\, p \in [P]\setminus \{1\},
\label{eq:depth_mip:mono}
\\
&\sum_{g\in \mathbb{G}_q} \sum_{p=1}^{P} z_{g,\,p} \cdot b_{p,\,d} \;\leqslant\; 1, \
\forall\, q \in [Q], \ d \in [D].
\label{eq:depth_mip:disjoint}
\end{align}
\end{subequations}
Constraint~\eqref{eq:depth_mip:assign} assigns each position to exactly one depth; \eqref{eq:depth_mip:mono} enforces a stepwise, monotonic depth ordering so that the depth at position $p{+}1$ is either equal to or exactly one greater than that at $p$; and \eqref{eq:depth_mip:disjoint} ensures that, at any depth, at most one gate acts on a given qubit $q$ (per-depth qubit disjointness). The bilinear products $z_{g,\,p}\cdot \,b_{p,\,d}$ in \eqref{eq:depth_mip:disjoint} are linearized exactly via the McCormick construction described in Section~\ref{subsec:linearization}.

\subsection{Target Constraint and Global-Phases}
The final condition that must be imposed is that the resulting circuit matches the target unitary $\widetilde T$, i.e., $\hat{G}_P = \mathcal{R}(e^{i \phi} \widetilde T)$ where $\phi \in [0,2\pi)$ is arbitrary. At first glance, this may appear to be a non-convex constraint, as it involves matching two unitaries up to a global-phase (GP). However, the following characterization of the constraint on the last cumulative unitary $\hat{G}_P$ is sufficient to ensure that the circuit reproduces the target up to an arbitrary global phase:
\begin{align}
&\text{GP Target:}\quad \hat{G}_P = \mathcal{R}((r + i s) \cdot \widetilde T), \label{eq:GP-condition}
\end{align}
where $r, s \in [-1,1]$ are newly introduced real variables that represent $e^{i\phi}$ in rectangular form. This constraint avoids the $\phi$-dependency in Eq.~\eqref{eq:GP-condition} while preserving solution equivalence up to global phase.

\begin{proposition}[]
\label{prop:complex_magnitude}
Let $\widetilde T\in\mathcal{U}(2^{Q})$ and let $\hat G_P\in\mathcal{U}(2^{Q+1})$ denote the cumulative product of the selected elementary unitary gates. If, in addition, the linear constraint in \eqref{eq:GP-condition} holds, then $|r+is|=1$ (equivalently $r^{2}+s^{2}=1$).
\end{proposition}

\begin{proof}
Using unitarity of $\hat{G}_P$ and $\widetilde T$, $
\mathbbm{1}_{2^{Q+1}}
=\hat{G}_P^\dagger \hat{G}_P
=\big(\mathcal{R}((r-is)\widetilde T^\dagger)\big) \big(\mathcal{R}((r+is)\widetilde T)\big)
=|r+is|^{2}\,\mathcal{R}(\widetilde T^\dagger \widetilde T)
=|r+is|^{2}\,\mathbbm{1}_{2^{Q+1}}.
$
Thus $\mathbbm{1}_{2^{Q+1}}=|r+is|^{2}\,\mathbbm{1}_{2^{Q+1}}$ implies $|r+is|^{2}=1$, i.e., $r^{2}+s^{2}=1$.    
\end{proof}

Consequently, when \eqref{eq:GP-condition} holds with $\widetilde T$ and $\hat{G}_P$ unitary, the non-convex identity $r^{2}+s^{2}=1$ follows automatically and need not be imposed explicitly. Thus, the matching of $\hat{G}_P$ to $\widetilde T$ up to a global phase is enforced entirely by the linear GP conditions. 

In principle, the correct condition is equality up to a global phase, as presented above. In practice, however, if both $\widetilde T$ and all generated circuits lie in $\mathcal{SU}(2^Q)$, one may impose the stricter condition with $r=1, \ s = 0$:
\begin{equation}\label{eq:exact_target}
\text{Exact Target:} \quad \hat G_P = \mathcal{R}(\widetilde T),
\end{equation}
which enforces exact equality in $\mathcal{SU}(2^Q)$. This ignores the residual action of the center $\mathcal Z_{2^Q} = \{ e^{2\pi i k / 2^Q} \mathbbm{1}_{2^{Q}}: k=0,\dots,2^Q-1 \}$, but since global phases are physically irrelevant, this approximation has no impact on the compiled circuit’s action on quantum states. We adopt this simplification in our numerical implementation for efficiency, with results presented in Section~\ref{subsec:gp-expts}.

\subsection{Compilation Objectives: Gate Count}
\label{subsec:objectives}
Using the variables introduced above, we define objective functions for common compilation goals.

\medskip
\noindent
\textbf{Objective: Hardware-aware gate count}\\

\noindent
Minimize a weighted count of non-identity gates:
\begin{equation}
\mathrm{minimize}\ \sum_{p=1}^{P}\ \sum_{g \in \mathbb{G}\setminus\{\mathbbm{1}\}} w_g\, z_{g,\ p},
\label{eq:obj_gate_count}
\end{equation}
where $w_g\!\ge\!0$ is a fixed per-gate weight parameter derived from hardware calibration or design priorities. Setting all $w_g\equiv 1$ recovers plain gate count; assigning larger $w_g$ to non-Clifford gates, especially $T$, targets fault-tolerant resource costs ($T$-count), while weighting multi-qubit/entangling gates more heavily reflects NISQ-era hardware coupler error rates. Binary values on $w_g$ recover subset counts (e.g., only non-Clifford or only multi-qubit), and excluding $\mathbbm{1}$ from the objective ensures unused positions incur no cost.

\medskip
\noindent
\textbf{Objective: Circuit depth} \\

\noindent
To minimize the circuit depth $D(\mathcal{C})$, we optimize the schedule by minimizing the depth at the final position $P$. This is achieved via the objective:
\begin{equation}
\mathrm{minimize} \ \sum_{d=1}^{D} d \cdot b_{P,\,d}.
\label{eq:depth_mip:objective}
\end{equation}
Because \eqref{eq:depth_mip} enforces a nondecreasing schedule that can increase by at most one per position, the circuit depth equals the depth of the final position, i.e., $D(\mathcal{C})$ as per Definition~\ref{defn:depth}. Minimizing this objective, therefore, yields a valid schedule with maximal parallelization of gates and hence minimal overall circuit depth.

\subsection{Compilation Objectives: Circuit Fidelity}
\label{subsec:approx}
Exact circuit synthesis is often infeasible and, in many applications, unnecessary: it typically suffices to approximate a target unitary $\widetilde T$ within a tolerance $\varepsilon>0$ in a phase–invariant metric (e.g., the fidelity). For any fixed, finite universal gate set, such as Clifford+$T$, the Solovay–Kitaev theorem guarantees that generic unitaries can be approximated to arbitrary precision with sequence length polylogarithmic in $1/\varepsilon$ (see, e.g., \cite{dawson_solovay-kitaev_2005}). In other settings, one may prefer alternative discrete families—such as Klein’s icosahedral group augmented by the “Super Golden $T$-gate”~\cite{parzanchevski_super-golden-gates_2018}—or seek to realize the canonical $H$, $X$, and $T$ gates by braiding or weaving Fibonacci anyons on topological quantum computers ~\cite{rouabah_compiling_2021,field_introduction_2018}. In these regimes, the design goal shifts: rather than minimizing gate count per se, one fixes a resource budget (total count, depth, $T$-count, etc.) and maximizes circuit fidelity. 

We quantify performance using the phase-invariant fidelity \cite{khatri2019quantum},
\begin{equation}\label{eq:fidelity-def}
F((U_p)_{p=1}^P)  \equiv \frac{1}{2^{2Q}}\left|\operatorname{Tr}\!\big(\widetilde T^\dagger \prod_{p=1}^{P} U_p\big)\right|^2.
\end{equation}
This is the squared, normalized Hilbert–Schmidt overlap between the implemented unitary $\prod_{p}U_p$ and the target $\widetilde T$. It obeys $0\leqslant F(\cdot) \leqslant 1$, with $F(\cdot)=1$ \emph{if and only if} $\prod_{p}U_p=e^{i\phi}\widetilde T$ for some global-phase $\phi$, and it is invariant under rephasing of either argument.

\medskip  
\noindent\textbf{Exact fidelity objective} \\ 
When MIP solvers support non-convex objectives, we optimize the phase-invariant fidelity \eqref{eq:fidelity-def} directly. Using real-encoding from Prop.~\ref{prop:Trace-identities}, define:  
\begin{align}  
\alpha &\equiv \tfrac{1}{2^{Q+1}}\operatorname{Tr}\!\big(\mathcal{R}(\widetilde T)^{\mathsf T}\hat{G}_P \big) = \Re\!\Big(\tfrac{1}{2^{Q}}\operatorname{Tr}(\widetilde T^\dagger\prod_{p=1}^{P} U_p)\Big), \label{eq:lin-fid-obj} \\  
\beta &\equiv -\tfrac{1}{2^{Q+1}}\operatorname{Tr}\!\big(J_{2^Q}\mathcal{R}(\widetilde T)^{\mathsf T}\hat{G}_P\big) = \Im\!\Big(\tfrac{1}{2^{Q}}\operatorname{Tr}(\widetilde T^\dagger \prod_{p=1}^{P} U_p)\Big), \label{eq:lin-fid-im}  
\end{align}  
where $J_{2^Q} = \mathbbm{1}_{2^Q} \otimes \begin{psmallmatrix} 0 & -1 \\ 1 & 0 \end{psmallmatrix}$. Since $F(\cdot) = \alpha^2 + \beta^2$, exact fidelity maximization requires  
\begin{equation}\label{eq:exact-fid-obj}  
\max\ \alpha^2 + \beta^2,  
\end{equation}  
while subject to all the base model's linear constraints. Owing to the modulus-squared trace in our decision variables, objective \eqref{eq:exact-fid-obj} is non-convex, which precludes direct treatment by convex optimization methods. While theoretically exact, this formulation becomes computationally intractable beyond small instances due to combinatorial non-convexity. We therefore investigate two linear, MILP-amenable surrogates of \eqref{eq:fidelity-def}. Their comparative empirical performance is reported in Section~\ref{sec:examples}.

\medskip  
\noindent\textbf{Linearized fidelity (real-part surrogate)}  \\
We now present a MILP-amenable surrogate by maximizing only the real trace overlap of the exact objective in \eqref{eq:exact-fid-obj}:  
\begin{equation}\label{eq:lin-fid-obj-again}  
\max\ \alpha 
\end{equation}  
Since $\beta^2 \geqslant 0$, $\alpha^2 \leqslant F(\cdot)$ provides a strict lower bound for $F(\cdot)$. Critically, this surrogate becomes asymptotically tight as $F(\cdot) \to 1$: when $F(\cdot) > 0.999$, residual phase errors ($\beta$) are typically negligible under $\mathcal{SU}(2^Q)$ encoding, ensuring $\alpha \approx \sqrt{F(\cdot)}$. 

\medskip
\noindent\textbf{Frobenius Error Objective} \\

\noindent
Next, we present an equivalent characterization of the linearized fidelity objective \eqref{eq:lin-fid-obj-again} while enabling scalable implementation by parametrizing the mismatch between the compiled circuit and the target unitary. Introduce a deviation matrix $E\in\mathbb{R}^{2^{Q+1}\times 2^{Q+1}}$ (elementwise bounds $-\epsilon\leqslant E_{ij}\leqslant \epsilon$) via
\begin{equation}\label{eq:approx-target}
\hat{G}_P \;=\; \mathcal{R}(\widetilde T) + E .
\end{equation}
A natural phase-invariant fidelity objective surrogate (in the real embedding) is to minimize the Frobenius error
\begin{equation}\label{eq:oa-frob}
\min\ \|E\|_F^2 \;=\; \min\ \operatorname{Tr}(E E^{\mathsf T}),
\end{equation}
which is a \textit{convex} quadratic objective with the existing base model's linear constraints. We now present Proposition~\ref{prop:objective-equivalence}, which characterizes the relationship between Frobenius-error minimization \eqref{eq:oa-frob} and fidelity maximization \eqref{eq:exact-fid-obj} objectives.
\begin{proposition}
\label{prop:objective-equivalence}
Consider a circuit compilation problem with fixed maximum depth $P$, fixed gate set $\mathbb{G}$, and target unitary $\widetilde{T}$. Let $\mathcal{S}$ denote the feasible set of circuits $(U_p)_{p=1}^P$ such that $\hat{G}_P \;=\; \mathcal{R}(\widetilde T) + E$. Then:
\begin{enumerate}
    \item[$\boldsymbol{(i)}$] $\|E\|_F^2 = 2^{Q+2}(1-\alpha)$, where $\alpha$ is defined as in~\eqref{eq:lin-fid-obj}.
    \item[$\boldsymbol{(ii)}$] Minimizing $\|E\|_F^2$ is equivalent to maximizing $\alpha$, the real trace overlap of the fidelity objective $F(\cdot)$, but \emph{not} necessarily to maximizing $F(\cdot)$ when $\max_{\mathcal{S}} F(\cdot) < 1$.
\end{enumerate}
\end{proposition}

\begin{proof}
\noindent$\boldsymbol{(i)}$: 
Using the orthogonality of $\hat{G}_P$ and $\mathcal{R}(\widetilde{T})$, and expanding the Frobenius norm using $E = \hat{G}_P - \mathcal{R}(\widetilde{T})$:
\begin{align*}
\|E\|_F^2 &= \operatorname{Tr}\big((\hat{G}_P - \mathcal{R}(\widetilde{T}))(\hat{G}_P - \mathcal{R}(\widetilde{T}))^{\mathsf{T}}\big) \\
&= 2 \cdot \operatorname{Tr}(\mathbbm{1}_{2^{Q+1}}) - \operatorname{Tr}(\hat{G}_P\mathcal{R}(\widetilde{T})^{\mathsf{T}}) - \operatorname{Tr}(\mathcal{R}(\widetilde{T})\hat{G}_P^{\mathsf{T}})
\end{align*}

\noindent
By trace properties (see (Prop.~\ref{prop:Trace-identities}).), both cross terms equal $\operatorname{Tr}(\mathcal{R}(\widetilde{T})^{\mathsf{T}}\hat{G}_P)$, and by equation \eqref{eq:lin-fid-obj}, so:
\begin{align*}
\|E\|_F^2 &= 2 \cdot \operatorname{Tr}(\mathbbm{1}_{2^{Q+1}}) - 2 \cdot \operatorname{Tr}(\mathcal{R}(\widetilde{T})^{\mathsf{T}}\hat{G}_P) \\
&= 2 \cdot 2^{Q+1} - 2 \cdot (2^{Q+1} \alpha) \\
&= 2^{Q+2}(1 - \alpha)
\end{align*}
This completes the proof of $\boldsymbol{(i)}$.

\noindent$\boldsymbol{(ii)}$:  From $\boldsymbol{(i)}$, $\|E\|_F^2 = 2^{Q+2}(1 - \alpha)$ is a strictly decreasing function of $\alpha$. Therefore, minimizing $\|E\|_F^2$ is equivalent to maximizing $\alpha$. However, since $F(\cdot) = \alpha^2 + \beta^2$, maximizing $\alpha$ does not necessarily maximize $F(\cdot)$ when $\beta$ can vary independently. This independence arises because the orthogonality constraint $\hat{G}_P\hat{G}_P^{\mathsf{T}} = \mathbbm{1}$ permits multiple values of $\beta$ for fixed $\alpha$ (and hence fixed $\|E\|_F^2$), depending on the alignment of $E$ with $J_{2^Q}\mathcal{R}(\widetilde{T})$.
\end{proof}

\begin{remark}
\label{rem:fidelity_bound}
From the identity $\boldsymbol{(i)}$ of Proposition \eqref{prop:objective-equivalence} and fidelity definition $F(\cdot) = \alpha^2 + \beta^2$, we obtain a tight lower bound:
\begin{equation}
\label{eq:fidelity-lb}
F(\cdot) \geqslant \left(1 - \frac{\|{E}\|_F^2}{2^{Q+2}}\right)^2,
\end{equation}
with equality iff $\beta = 0$. This bound is asymptotically tight as $F(\cdot) \to 1$: when $F(\cdot) > 0.999$, residual phase errors ($\beta$) can typically become negligible under $\mathcal{SU}(2^Q)$ encoding. 
\end{remark}
Although \eqref{eq:oa-frob} is convex quadratic, to retain a \textit{linear} objective, and thus a MILP formulation that scales, we globally under-approximate each quadratic term $E_{ij}^2$ by a finite family of first-order (tangent) approximations to $x\mapsto x^2$ on an uniform grid $\{a_k\}_{k=1}^{K}\subset[-\epsilon,\epsilon]$ (see Table~\ref{tab:approx-fib-3targets} for speedups). Introduce auxiliaries $\widehat{E}_{ij}\!\ge\!0$ and write the objective with tangent constraints as
\begin{equation}\label{eq:oa-tangents}
\begin{aligned}
\min\ & \sum_{(i,j)\in[2^{Q+1}]\times[2^{Q+1}]} \ \widehat{E}_{ij}\\
\text{s.t.}& \quad \widehat{E}_{ij} \;\geqslant\; 2 a_k\, E_{ij} - a_k^2,\quad \forall \ (i,j),\ k\in[K].
\end{aligned}
\end{equation}
Consequently, for any feasible solution of \eqref{eq:oa-tangents}, the optimal objective value is a valid \emph{lower bound} to the true quadratic objective \eqref{eq:oa-frob}. Tightness, also a surrogate for Fidelity, improves monotonically with $K$. 

Alternative objectives based on other matrix norms (e.g., $\ell_1$ and $\ell_\infty$) were evaluated against the Frobenius-based quadratic outer-approximation surrogate for measuring compilation error. Across benchmark circuits, they yielded negligible fidelity gains while increasing wall-clock time; accordingly, we report results using the quadratic outer-approximation metric, which offered the best quality–time trade-off.

Neither linear surrogate (\eqref{eq:lin-fid-obj-again} or \eqref{eq:oa-tangents}) guarantees optimality for fidelity maximization in \eqref{eq:fidelity-def}, but both achieve near-optimal fidelities (\(F(\cdot) > 0.999\)) significantly faster than non-convex or brute-force methods. Phase-invariant alternatives (e.g., minimizing \(\|\hat{G}_P - e^{i\phi}T\|_F\) over \(\phi\)) yield equivalent linear formulations. We evaluate the scalability of these approaches versus the exact MIQP \eqref{eq:exact-fid-obj} in Section~\ref{sec:examples}.

\section{Speeding up the Compilation}
\label{sec:valid_constraints}

The MILP formulation in Section \ref{sec:math_form} (the ``base'' formulation), exactly models the quantum circuit synthesis problem, theoretically guaranteeing global optimality for any transpilation task with a finite gate set. However, the combinatorial structure of quantum circuits induces significant computational challenges: the formulation exhibits inherent symmetry (multiple equivalent gate sequences yielding identical unitaries), a weak continuous (linear-programming) relaxation due to the bilinear nature of quantum operations, and exponential growth in variable count with respect to qubit count. These properties lead to slow convergence in branch-and-cut algorithms, as evidenced by empirical results showing impractical solve times for circuits exceeding 3 qubits or depth 10.

To strengthen the formulation, we derive three families of valid constraints that preserve the global optimal solution while improving the integrality gap and reducing symmetry in the branch-and-bound tree. While theoretical analysis confirms their validity, the practical efficacy of these enhancements depends critically on the specific circuit and must be evaluated empirically. Our experimental results in Section \ref{sec:examples} demonstrate that these domain-specific valid constraints reduce solve times in state-of-the-art solvers by up to two orders of magnitude for non-trivial benchmarks.

\subsection{Symmetry Constraints}
\label{subsec:symm_constraints}
The first set of valid constraints works by reducing the search space by eliminating valid symmetrical optimal solutions. One immediate example arises from the placement of identity gates. Given an optimal compilation with identity gates, an equally optimal solution can be obtained by shifting these identities to other positions. To eliminate such symmetry, we introduce the constraint:
\begin{equation}\label{eq:IGS-constraint}
z_{\mathbbm{1},\,p-1} \leqslant  z_{\mathbbm{1},\,p}, \quad \forall\, p \in [P] \setminus \{1\}.
\end{equation}
We refer to this as the \textit{Identity Gate Symmetry} constraint.

Similarly, consider two gates $g_i$ and $g_j$ that commute with each other. If these gates appear consecutively, swapping their order yields another equally optimal solution. To break this symmetry, we impose:
\begin{equation}
\label{eq:commuting_pair}
z_{g_i,\, p-1}+z_{g_j,\, p}\leqslant  1, \quad \forall\, p \in [P] \setminus \{1\}.
\end{equation}
thereby reducing the search space by eliminating similar solutions.

A generalization of this commutativity constraint involves evaluating pairs of gates and identifying equivalent products. The problem is then constrained so that only one of the equivalent pairs is allowed, while the other is made infeasible. For example, this well-known relationship of the Pauli gates $X$ and $Z$ together with the Hadamard: $X \cdot H = H \cdot Z$ when all act on the same qubit, they are equivalent. By constraining the possibility that $Z$ can follow a Hadamard, i.e., $z_{H,\, p}+z_{Z,\, p+1}\leqslant  1$, we can reduce the symmetries in the search space of optimal solutions.

Similarly, symmetry reduction via \emph{equivalent triplet} elimination is crucial when distinct three-gate sequences implement the same unitary. This approach is particularly critical for topological quantum compilation via Fibonacci anyon braiding \cite{hormozi2007topological}, where the braid group $B_n$ (governing exchanges of $n$ anyons with $n-1$ generators $\sigma_i$) satisfies the Yang–Baxter relation \cite{baxter2016exactly, peng2022quantum}: 
$$
\sigma_i \cdot \sigma_{i+1}  \cdot \sigma_i = \sigma_{i+1}  \cdot \sigma_i  \cdot \sigma_{i+1}, \qquad \forall \ i=1,\dots,n-2,
$$
which induces local two-way redundancies. To prune symmetric branches without excluding optimal solutions, we enforce a canonical choice by forbidding one representative. Concretely, for consecutive positions ($p,\ p{+}1,\ p{+}2$) we impose the valid inequality
\begin{align}
&z_{\sigma_i,\, p}+z_{\sigma_{i+1},\,p+1}+z_{\sigma_i,\, p+2}\leqslant  2, \nonumber\\ 
&\qquad \quad \forall i\in \{1,\dots,n-2\},\ \  \forall, p\in \{1,\dots,P-2\},
\end{align}
which excludes ($\sigma_i \cdot\sigma_{i+1} \cdot\sigma_i$) while retaining feasibility through the equivalent ($\sigma_{i+1} \cdot\sigma_i \cdot\sigma_{i+1}$). For Fibonacci anyon models where such local symmetries proliferate with $n$, this symmetry-breaking technique significantly accelerates MIP's branch-and-bound convergence while preserving optimality.

\begin{remark}
\label{rem:depth_no_commute_cons}
Eliminating one representative of two locally equivalent canonical patterns (e.g., commuting gates or triplet sequences) is valid for \emph{order-insensitive} objectives (total gate count \eqref{eq:obj_gate_count}; phase-invariant fidelity \eqref{eq:exact-fid-obj}, \eqref{eq:oa-frob}), but generally \emph{invalid} for depth minimization objective in \eqref{eq:depth_mip:objective}, where gate order affects parallelism via \eqref{eq:depth_mip:mono}--\eqref{eq:depth_mip:disjoint}.

\emph{Counterexample.} Define $T_1 := T\otimes \mathbbm{1}_2,\qquad T_2 := \mathbbm{1}_2\otimes T,\qquad \mathrm{CNOT}_{1,2}\in \mathcal{U}(2)$ with (control on qubit 1, target on qubit 2). Since $T_1$ is diagonal and acts only on the control, $T_1$ commutes with $\mathrm{CNOT}_{1,2}$, hence
\[
T_1  \cdot \mathrm{CNOT}_{1,2}  \cdot T_2 = \mathrm{CNOT}_{1,2}  \cdot (T \otimes T).
\]
Without symmetry breaking, schedule $\mathrm{CNOT}_{1,2}$ at depth 1 and $(T \otimes T)$ in parallel at depth 2 is ``optimal''. If a symmetry-breaking rule enforces $T_1\prec \mathrm{CNOT}_{1,2}$, then $T_1$ is at depth 1, $\mathrm{CNOT}_{1,2}$ at depth 2, and $T_2$ must be at depth 3. Thus the enforced local symmetry elimination strictly increases optimal depth; such symmetry constraints are \emph{invalid} for depth minimization.
\end{remark}

\subsection{Redundancy-elimination Constraints}

To eliminate redundant gate sequences while maintaining computational tractability, we impose valid constraints exclusively for sequences of length \(2 \leqslant k \leqslant  5\). For a gate set \(\mathbb{G}\) and any sequence \((g_1, \dots, g_k) \in \mathbb{G}^k\) satisfying the algebraic identity \(g_1 \cdots g_k = g'\) for some \(g' \in \mathbb{G}\), the following constraint is enforced for all starting positions $p\leqslant  P-k+1$:  
\begin{align}
\sum_{i=1}^{k} z_{g_i,\, p+i-1} \leqslant  k-1, \quad \forall k \in [5] \setminus \{1\}.
\end{align}  
This prevents the inclusion of the redundant sequence \((g_1, \dots, g_k)\) at consecutive positions \(p\) to \(p+k-1\), as its effect is algebraically equivalent to the single gate \(g'\).

We consider $k\leqslant 5$ (up to quintuplets) to cover the standard minimal library $\{\mathrm{CNOT},H,T\}$, identities canonical relations such as control–target reversal $(H\otimes H) \cdot\,\mathrm{CNOT}_{1,2}\,\cdot (H\otimes H)=\mathrm{CNOT}_{2,1}$, $T^2=S$, $X^2=\mathbb{I}$, $X\cdot H \cdot Z=H$, $H_2\cdot\mathrm{CNOT}_{1,2} \cdot H_2=\mathrm{CZ}_{1,2}$, as well as cyclic relations of weaving operators in topological quantum computation with Fibonacci anyon models \cite{simon_topological_2006, rouabah_compiling_2021}.

To summarize, symmetry constraints eliminate equivalent optimal circuits by identifying equivalence classes  and restricting to a canonical representative, while redundancy constraints remove provably dominated solutions that correspond to non-minimal gate sequences implementing identical unitaries. The combined application of both constraint classes induces a valid restriction of the feasible region of discrete solutions that preserves solution optimality while substantially reducing the branch-and-bound search space through elimination of both symmetric solutions and dominated sequences.

\subsection{Hermitian-conjugate Constraints}
\label{subsec:hermitian}

The constraints introduced in this section—\emph{Hermitian-conjugate constraints (HC)}—do not exclude any integer-feasible compilation; they encode algebraic identities that every valid compiled circuit already satisfies. As noted in Section~\ref{subsec:milp}, MILP solvers proceed via LP relaxations in which gate-selection variables may take fractional values; at such fractional points these identities generally fail. Enforcing them therefore leaves the convex hull of integer solutions unchanged while cutting off infeasible fractional solutions, yielding a tighter LP relaxation. A stronger relaxation improves the root-node dual bound, reduces the integrality gap, and typically shrinks the branch-and-bound tree. Consequently, these identities serve as valid constraints that accelerate convergence.

The following Hermitian-conjugate \eqref{eq:hc_gamma} identity, applied near the end of the circuit, was found significantly effective for the circuit compilation problem as modeled in this paper. Given the recursion relation in \eqref{eq:G-hat} and the exact target relation \(\hat{G}_P=\mathcal{R}(\widetilde T)\), for any integer \(1\leqslant \gamma\leqslant P\),
\begin{equation}
\hat{G}_{P-\gamma}
=\mathcal{R}(\widetilde T)\left(\prod_{t=P-\gamma+1}^{P} G_t \right)^{\dagger}.
\tag{HC--$\gamma$}
\label{eq:hc_gamma}
\end{equation}
This follows directly from unitarity: right-multiplying by the Hermitian conjugate of the trailing product of gates yields the identity. 

In this work we enforce HC--1 (\(\gamma=1\)) and HC--2 (\(\gamma=2\)). Because the gate-selection constraint \eqref{eq:gate_select} chooses exactly one \(g\in\mathbb{G}\) at each position, we have
$
G_t^{\dagger} \;=\; \sum_{g\in\mathbb{G}} z_{g,t}\,\mathcal{R}(g)^{\dagger}.
$
Therefore HC--1 is linear in the decision variables \(z_{g,t}\). By contrast, HC--\(\gamma\) with \(\gamma\geqslant 2\)) involves products, which are bilinear in \(\{z_{g,t}\}\); depending on solver capabilities, these constraints can be imposed directly or via standard convexifications (see Section~\ref{subsec:linearization}). \\ \\

\medskip
\noindent
\textbf{Global-phase generalization of HC-1 identity} \\

\noindent
The Hermitian-conjugate conditions must be generalized to reflect the fundamental fact that unitaries are defined only up to a global phase, as in \eqref{eq:GP-condition}. Although this generalization is nontrivial for arbitrary \(\gamma\), we now present the Hermitian-conjugate identity HC--1, whose corresponding phase-invariant version is:
\begin{align}
\hat{G}_{P-1} \ = \ \mathcal{R}(r+is) \cdot \sum_{g\in\mathbb{G}} z_{g,P}\cdot\mathcal{R}(\widetilde{T} \cdot g^{\dagger}),
\label{eq:HC_global}    
\end{align}
where let \(\widetilde{T}\cdot g^{\dagger}=A_g+iB_g\) for fixed matrices \(A_g, B_g \in \mathbb{R}^{2^Q \times 2^Q}\) whose entries satisfy \(|(\widetilde{T}\cdot g^{\dagger})_{ij}| \leqslant 1\) since both $\widetilde{T}$ and $g^{\dagger}$ are unitary.

The product of \((r,s)\) with the binary variables \(\{z_{g,P}\}\) renders \eqref{eq:HC_global} \emph{nonlinear}. We eliminate this nonlinearity via an exact disjunctive reformulation. Let \(\mathbb{J}_{2^Q}\) denote the \(2^Q\times 2^Q\) all-ones matrix. Then the phase-sensitive condition in \eqref{eq:HC_global} is exactly equivalent to the entry-wise linear inequalities stated in the following Proposition \ref{prop:hc1_GP_claim}. Here \(\lvert\cdot\rvert\) denotes the matrix whose \((i,j)\)-th entry is the absolute value of the difference between the corresponding matrix entries, and \(\mathcal{R}^{-1}(\cdot)\) is the inverse real-to-complex map defined in Section \ref{subsec:real-encoding}.

\begin{proposition}
\label{prop:hc1_GP_claim}
Let \(\mathbb{G}\) be a finite set of unitary gates. Suppose \(r, s \in [-1,1]\) satisfy \(r^2 + s^2 = 1\), and \(|(\widetilde{T} \cdot g^{\dagger})_{ij}| \leqslant 1\) for all \(i,j\). Then the nonlinear constraint in \eqref{eq:HC_global} is equivalent to the set of element-wise linear inequalities for all \(g \in \mathbb{G}\):  
\begin{align*}
&\big|\operatorname{Re}\big(\mathcal{R}^{-1}(\hat{G}_{P-1})\big) - (r A_g - s B_g)\big| \leqslant 2\cdot (1 - z_{g,P})\cdot \mathbb{J}_{2^Q}, \\
&\big|\operatorname{Im}\big(\mathcal{R}^{-1}(\hat{G}_{P-1})\big) - (r B_g + s A_g)\big| \leqslant 2\cdot (1 - z_{g,P})\cdot \mathbb{J}_{2^Q},
\end{align*}  
\end{proposition}

\begin{proof}
Assume first that the nonlinear constraint \eqref{eq:HC_global} is valid. Then there exists a unique \(g^* \in \mathbb{G}\) with \(z_{g^*,P} = 1\) and \(z_{g,P} = 0\) for \(g \ne g^*\), so \(\hat{G}_{P-1} = \mathcal{R}(r+is) \cdot \mathcal{R}(\widetilde{T} \cdot (g^*)^{\dagger})\). By the properties of the real encoding \(\mathcal{R}\) (from Section \ref{subsec:real-encoding}), this implies \(\mathcal{R}^{-1}(\hat{G}_{P-1}) = (r+is)(\widetilde{T} \cdot (g^*)^{\dagger})\), and thus \(\operatorname{Re}(\mathcal{R}^{-1}(\hat{G}_{P-1})) = r A_{g^*} - s B_{g^*}\) and \(\operatorname{Im}(\mathcal{R}^{-1}(\hat{G}_{P-1})) = r B_{g^*} + s A_{g^*}\). For \(g = g^*\), the right-hand sides of the linear constraints vanish, and the aforementioned equalities hold exactly. For \(g \ne g^*\), \(z_{g,P} = 0\) gives a right-hand side of \(2 \cdot \mathbb{J}_{2^Q}\). Since \(|(\widetilde{T} \cdot g^{\dagger})_{ij}| \leqslant 1\) and \(r^2 + s^2 = 1\) is implied due to Proposition \ref{prop:complex_magnitude}, the modulus of each entry of \((r+is)(\widetilde{T} \cdot g^{\dagger})\) is bounded by 1. Thus, both \(\mathcal{R}^{-1}(\hat{G}_{P-1})\) and \((r+is)(\widetilde{T} \cdot g^{\dagger})\) have real and imaginary parts with entries in \([-1,1]\), making the element-wise differences bounded by 2. Conversely, if the linear constraints hold with a unique \(g^*\) having \(z_{g^*,P} = 1\), then for \(g^*\) the constraints enforce exact equality, recovering the nonlinear constraint, while for other \(g\) the constraints are automatically satisfied due to the 2-bound. Thus, the reformulation is exact.
\end{proof}

\begin{figure*}[t]
\centering
\includegraphics[width=\textwidth]{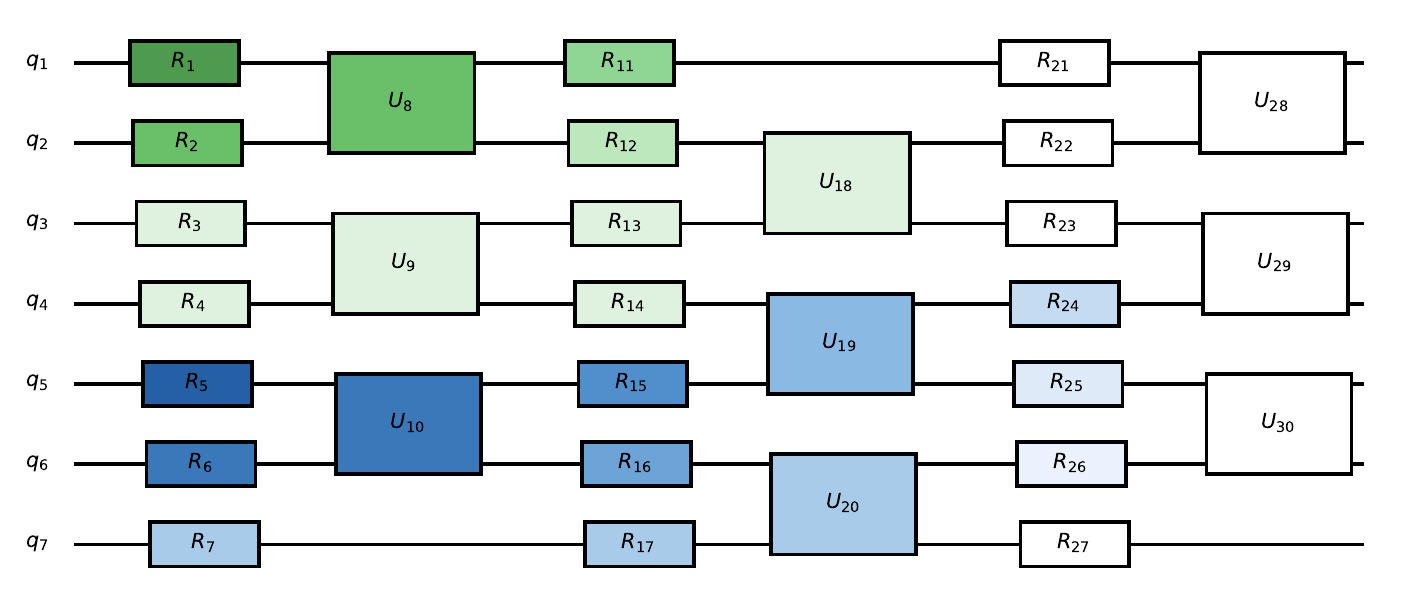}
\caption{Hardware-efficient brickwork ansatz on $7$ qubits (cf.\ \cite{kandala_hardware-efficient_2017,leone_practical_2024}), illustrating the identification of optimization blocks in the rolling-horizon optimization (RHO). Single- and two-qubit gates are slightly offset horizontally to indicate that, while typically parallel, any fixed virtual time ordering is admissible; for exposition, we adopt a top-to-bottom order. Green shading denotes the gate insertion order in the \emph{first} rolling-horizon window (dark$\to$light); blue shading analogously denotes the \emph{second} window.}
\label{fig:rho_brickwork}
\end{figure*}

\section{Rolling-Horizon Optimization Framework}

The computational complexity of exact quantum circuit synthesis via MILP exhibits exponential dependence on qubit count $Q$, with formulation size scaling as $\Theta(4^Q)$. This inherent intractability renders global optimization impractical beyond modestly sized circuits, despite the inclusion of valid strengthening constraints above.

To address this complexity barrier while preserving solution quality, \texttt{QCOpt} employs a rolling-horizon optimization (RHO) wrapper for the specific task of gate-count minimization. The global problem is decomposed into a sequence of windows (horizons) of subproblems that are solved independently (with a potential for parallelization); their solutions are then concatenated to yield a hardware-feasible, near-optimal circuit. This improves tractability at the cost of forfeiting global optimality guarantees.

Traditional rolling-horizon methods advance strictly along the circuit’s time axis, optimizing gates in chronological order. Our RHO generalizes this by rolling the horizon over subsets of qubits (with short temporal windows), thereby partitioning both temporal and spatial dimensions and reducing complexity.

RHO in our setting presupposes an existing (typically suboptimal) target circuit to refine. That is, given an initial gate list that implements the full target $\widetilde T$ (e.g., from a heuristic transpiler), we do not synthesize from scratch; instead, we iteratively improve contiguous subcircuits while preserving overall functionality. The two design levers are the \emph{window–selection policy} (how we choose the next block to optimize) and the \emph{acceptance policy} (how much of the optimized block we lock in before rolling forward). In practice, the bottleneck step is selecting the first optimization block: this choice determines how the horizon will roll and how quickly the qubit set grows.

After optimizing the first block, we accept only an initial prefix of that optimized window (the acceptance window) and then retarget the remaining problem by effectively left–multiplying with the Hermitian-conjugate of what we accepted. If $B_{\text{acc}}$ denotes the product of the accepted gates (in forward time order), the full target is updated as
\begin{equation}
\widetilde T \;\leftarrow\; B_{\text{acc}}^{\dagger}\, \widetilde T.
\end{equation}
For example, if in the first iteration we accept a single $\mathrm{CNOT}_{1,2}$, then the new full target becomes $\mathrm{CNOT}_{1,2} \cdot \widetilde T$. Setting the acceptance window equal to the optimization window eliminates overlap between chunks but tends to hurt quality; accepting only a minimal prefix improves quality at the expense of more iterations.

The first block we seek should satisfy several properties. (i) It must involve at most $\texttt{max\_qubits}$ distinct qubits to keep the effective Hilbert space small. (ii) Its length must not exceed $\texttt{window\_length}$. (iii) It should \emph{roll primarily in time}: we extend the window chronologically and only enlarge the qubit set when forced to by multi–qubit gates, thereby avoiding unnecessary growth in subsystem size. (iv) It should satisfy a \emph{closure} condition on included qubits: once a gate on qubit $q$ enters the window up to index $i$, all gates acting on $q$ that occur before $i$ must also be included. This prevents inconsistent partial histories on any involved qubit and yields well–posed subproblems.

Operationally, our RHO policy: finds the earliest feasible block that obeys the cardinality bound on qubits, the length bound on gates, the temporal preference (extend in time first, add qubits only when necessary), and the per–qubit closure rule; optimizes that block; accepts a prefix; retargets the remainder; and repeats. This policy balances stability (small qubit footprint with preserved local context) with progress (steady forward motion in time) and underlies the routines with four key functions described in detail below:

\paragraph{\texttt{gates\_on\_qubits\_up\_to}} This function takes a target gate sequence $t$, a set of qubits $q$, and an index $i$. It selects all gates from the beginning of the sequence up to position $i$ that act on any qubit in $q$. The function iterates through the sequence, checking the qubits each gate acts upon, and appends matching gates to a list.

\begin{figure}[htbp]  
\centering
\begin{minipage}{\columnwidth}
\begin{algorithm}[H]
\caption*{\textsc{gates\_on\_qubits\_up\_to}(\texttt{t}, \texttt{q}, \texttt{i})}
\begin{algorithmic}[1]
\State \textbf{Input:} gate list, qubit set, index
\State \textbf{Output:} gates acting on  \texttt{q} up to $\texttt{t}[\texttt{i}]$
\State \texttt{selected\_gates} $\gets$ [ ]
\For{$\texttt{j} = 1$ to \texttt{i}}
    \State \texttt{gate} $\gets \texttt{t}[\texttt{j}]$
    \State \texttt{gate\_qubits} $\gets$ \Call{ExtractQubits}{\texttt{gate}}
    \If{\Call{Intersection}{\texttt{gate\_qubits}, \texttt{q}} $\neq \emptyset$}
        \State \Call{Append}{\texttt{selected\_gates}, \texttt{gate}}
    \EndIf
\EndFor
\State \Return \texttt{selected\_gates}
\end{algorithmic}
\end{algorithm}
\end{minipage}
\end{figure}

\paragraph{\texttt{recursive\_gates\_on\_qubits\_up\_to}} Building on the previous function, this recursively expands the set of selected qubits. Initially calling \texttt{gates\_on\_qubits\_up\_to}, it updates the qubit set $q$ based on qubits involved in selected gates and repeats until no new qubits are found in order to fulfill the closure condition.

\begin{figure}[htbp]
\centering
\begin{minipage}{\columnwidth}
\begin{algorithm}[H]  
\caption*{\textsc{recursive\_gates\_on\_qubits\_up\_to}(\texttt{t}, \texttt{q}, \texttt{i})}  
\begin{algorithmic}[1]
\State \textbf{Input:} gate list, qubit set, index
\State \textbf{Output:} extended gate list acting on qubits up to $\texttt{t}[\texttt{i}]$
\State \texttt{selected\_gates} $\gets$ \Call{gates\_on\_qubits\_up\_to}{\texttt{t}, \texttt{q}, \texttt{i}}
\State \texttt{new\_q} $\gets$ \Call{ExtractQubits}{\texttt{selected\_gates}}
\If{$\Call{length}{\texttt{new\_q}} > \Call{length}{\texttt{q}}$}
    \State \Return \Call{recursive\_gates\_on\_qubits\_up\_to}{\texttt{t}, \texttt{new\_q}, \texttt{i}}
\Else
    \State \Return \texttt{selected\_gates}
\EndIf
\end{algorithmic}
\end{algorithm}
\end{minipage}
\end{figure}

\paragraph{\texttt{find\_first\_block}} Identifies the first feasible block of gates to optimize, constrained by a specified window length and maximum qubit count. It incrementally evaluates sequences until one exceeds the constraints, returning the largest feasible block found.

\begin{figure}[htbp]
\centering
\begin{minipage}{\columnwidth}
\begin{algorithm}[H]
\caption*{\textsc{find\_first\_block}(\texttt{t}, \texttt{window\_length}, \texttt{max\_qubits})}
\begin{algorithmic}[1]
\State \textbf{Input:} target gate sequence, window size, max number of qubits in a window
\State \textbf{Output:} best gate block satisfying constraints
\State $\texttt{q} \gets \Call{ExtractQubits}{\texttt{t}[1]}$
\State $\texttt{i} \gets 1$
\State \texttt{best\_sequence} $\gets$ [ ]
\While{true}

    \State \texttt{seq} $\gets$ \Call{recursive\_gates\_on\_qubits\_up\_to}{\texttt{t}, \texttt{q}, \texttt{i}} 
     \State \texttt{q} $\gets$ \Call{ExtractQubits}{\texttt{seq}}
    \State cond1 $\gets \texttt{i} > \Call{Length}{\texttt{t}}$
    \State cond2 $\gets \Call{Length}{\texttt{seq}} > \texttt{window\_length}$
    \State cond3 $\gets \Call{Length}{\texttt{q}} > \texttt{max\_qubits}$
    
    \If{cond1 \textbf{ or } cond2 \textbf{ or } cond3}
    \State \Return \texttt{best\_sequence}

    \ElsIf{$\Call{Length}{\texttt{seq}} = \texttt{window\_length}$}
        \State \Return \texttt{seq}
    \Else
        \State \texttt{best\_sequence} $\gets$ \texttt{seq}
        \State $i \gets i + 1$
    \EndIf
\EndWhile
\end{algorithmic}
\end{algorithm}
\end{minipage}
\end{figure}

\paragraph{\texttt{rolling\_horizon}} This function drives optimization by sequentially selecting and optimizing blocks, updating the gate sequence until all gates are processed.

\begin{figure*}[htbp]
\centering
\begin{minipage}{\textwidth}
\begin{algorithm}[H]
\caption*{\textsc{rolling\_horizon}($\texttt{target\_sequence}$, \texttt{window\_length}, \texttt{max\_qubits}, \texttt{elementary\_gates}, \texttt{accept\_window}, \texttt{optimizer})}
\begin{algorithmic}[1]
\State \textbf{Input:} target gate sequence, window size, max number of qubits in a window, elementary gate set, optimizer, accept window size
\State \textbf{Output:} optimized gate sequence
\State \texttt{output\_gates} $\gets$ [ ]
\State \texttt{full\_target} $\gets$ \texttt{target\_sequence}
\While{\texttt{full\_target} $\neq \emptyset$}
    \State \texttt{current\_target} $\gets$ \Call{find\_first\_block}{\texttt{full\_target}, \texttt{window\_length}, \texttt{max\_qubits}}
    \State \texttt{current\_qubits} $\gets$ \Call{ExtractQubits}{\texttt{current\_target}}
    \State \texttt{full\_target} $\gets$ \Call{Remove}{\texttt{full\_target}, \texttt{current\_target}}
    \If{$\Call{Length}{\texttt{current\_qubits}} \leqslant  1$ \textbf{or} $\Call{Length}{\texttt{current\_target}} = 1$ }
        \State \Call{Append}{\texttt{output\_gates}, \texttt{current\_target}}
    \Else
        \State \texttt{optimized} $\gets$ \Call{\texttt{QCOpt}}{\texttt{current\_target}, \texttt{elementary\_gates}, \texttt{optimizer}}

        \If{$\Call{Length}{\texttt{full\_target}} = 0$}
            \State \Call{Append}{\texttt{output\_gates}, \texttt{optimized}[:]}
        \Else
            \If{$\Call{Length}{\texttt{optimized}} >= \texttt{accept\_window}$}
                \State \Call{Append}{\texttt{output\_gates}, \texttt{optimized[:\texttt{accept\_window}]}}
                \State \texttt{full\_target} $\gets$ \Call{Append}{\texttt{optimized[\texttt{accept\_window}+1:]}, \texttt{full\_target}}
            \Else
                \State \Call{Append}{\texttt{output\_gates}, \texttt{optimized}}
            \EndIf
        \EndIf
    \EndIf
\EndWhile
\State \Return \texttt{output\_gates}
\end{algorithmic}
\end{algorithm}
\end{minipage}
\end{figure*}

To illustrate how the first block is found under rolling horizon, we use the 7-qubit hardware-efficient ansatz in Fig.~\ref{fig:rho_brickwork} (cf.~\cite{kandala_hardware-efficient_2017,leone_practical_2024}). Although single-qubit rotations are often treated as parallel, in practice, any fixed virtual time ordering is admissible; in the figure, we intentionally offset them slightly to make this explicit. In this example, we adopt a simple acceptance policy—accept the entire optimized window—and set the window length to $12$ and per-window qubit cap to $4$.

The search for the first block begins at the earliest gate, $R_{1}$ on qubit 1 (dark green). The next gate on that qubit is $U_{8}$; inserting $U_{8}$ triggers the closure condition, which simultaneously pulls in $R_{2}$, so both enter the window. The algorithm then adds $R_{11}$ and $R_{12}$ (now 5 gates). The next candidate is $U_{18}$, but closure requires adding $R_{3},R_{4},U_{9},R_{13},R_{14}$, bringing the total to 11. The subsequent candidate $U_{19}$—together with $R_{5},R_{6},U_{10},R_{15},R_{16}$ required by closure—would exceed both the qubit cap ($>4$) and the window length ($>12$), so these additions are rejected. The algorithm therefore optimizes the valid 11-gate window by forming the corresponding target and solving the exact MILP to minimize gate count. In Fig.~\ref{fig:rho_brickwork}, progressively lighter shades of green indicate the order in which gates enter this first window.

Because the acceptance policy fixes each optimized window, that portion of the target circuit is removed, and the search moves to the next “first block.” It resumes at $R_{5}$, now the earliest remaining gate. The second block is assembled as follows: $U_{10}$ and $R_{6}$ are added (closure), then $R_{15}$ and $R_{16}$; adding $U_{19}$ enlarges the active-qubit set to qubits 4–6, and adding $U_{20}$ pulls in $R_{7}$ and $R_{17}$, activating qubit 7. Subsequent rounds add $R_{24}$, $R_{25}$, and $R_{26}$. When the window hits its 12-gate cap, only those gates are optimized. The construction of this second block appears in progressively lighter shades of blue in Fig.~\ref{fig:rho_brickwork}. This simple example makes the rolling-horizon mechanics concrete: select the earliest feasible block (respecting length, qubit cardinality, and per-qubit closure), optimize, accept, retarget, and repeat.

\section{Numerical Experiments}
\label{sec:examples}

All MILP/MIQP formulations and the rolling-horizon algorithm are implemented in the open-source package \texttt{QuantumCircuitOpt.jl}. We omit usage details here; complete documentation and a user guide are available at \url{https://github.com/harshangrjn/QuantumCircuitOpt.jl}. The illustrative examples presented in this section demonstrate typical transpilation use cases and provide empirical evidence motivating several implementation choices. All experiments were run on a single machine with the following setup: \emph{CPU: \texttt{AMD EPYC 7R13} (32 cores / 64 threads, base 2.56\,GHz); RAM: \texttt{128\,GB}; OS: \texttt{Windows Server 2022 (21H2)}; solver: \texttt{Gurobi v12.0.2} (default settings), threads: \texttt{32}; Julia: \texttt{1.11.5}}.

Unless stated otherwise, all runs use the total gate-count objective \eqref{eq:obj_gate_count} with the baseline MILP constraints \eqref{eq:gate_select}–\eqref{eq:G-hat}. We denote by ``best-MILP'' the baseline augmented with all additional constraints—symmetry, redundancy–elimination, and the Hermitian-conjugate (HC-1) constraints (see Sec. \ref{sec:valid_constraints}). Reported run-time statistics may vary with hardware (e.g., core count) due to the parallelism of commercial MIP solvers (e.g., Gurobi) and may also depend on the solver choice. Further details on the number of qubits, elementary gate set size, and the imposed maximum circuit length for each example are provided in Appendix~\ref{app:example-params}.

\subsection{Performance of Global-Phase Conditions}\label{subsec:gp-expts}

In Section~\ref{sec:math_form} we introduced the global-phase equivalence target condition (Eq.~\eqref{eq:GP-condition}), which allows $\phi \in [0, 2\pi)$. We show the implementation of this formulation using the total gate-count objective under the ``best-MILP'' configuration, with Hermitian-conjugate (HC-1) constraints implemented in the global-phase perspective, as described in Proposition \ref{prop:hc1_GP_claim}. Table~\ref{tab:gp-compare-reduced} reports wall-clock times; the ``Exact'' column corresponds to the strict $\phi=0$ condition (Eq.~\eqref{eq:exact_target}). The exact condition is, on average, 28.3\% faster than global-phase equivalence. The gap likely stems from the larger feasible region in the global-phase formulation, which introduces combinatorial symmetries and increases branch-and-bound depth, degrading MILP solver performance.
Nevertheless, when all elementary gates are encoded in $\mathcal{SU}(2^Q)$, the exact target condition (Eq.~\eqref{eq:exact_target}) mostly outperforms the global-phase representation, and we therefore adopt it for all subsequent examples.

\begin{table}[t]
\centering
\small
\caption{Execution time (seconds) for two separate conditions: global-phase equivalence ($\phi \in [0, 2\pi)$) vs. exact ($\phi=0$).}
\label{tab:gp-compare-reduced}
\begin{tabular}{lrr}
\toprule
Target & Global-phase (s) & Exact (s) \\
\midrule
Controlled-$\sqrt{\mathrm{X}}$                              & 17.01    & 12.93 \\
Controlled-Hadamard                                         & 0.46  & 3.92  \\
Magic\footnotemark~\cite{vatan_optimal_2004}                & 10.26   & 3.35   \\
iSwap                                                       & 8.57  & 2.39  \\
single excitation Hadamard~\cite{arrazola_universal_2022}   & 8.30   & 5.26  \\
\midrule
$\mathrm{CNOT}_{1,3}$                                       & 16.52 & 2.34  \\
Fredkin~\cite{smolin_five_1996}                             & 41.58 & 43.30  \\
Miller~\cite{hung_optimal_2006}                             & 71.09 & 117.7 \\
Relative Toffoli~\cite{maslov_advantages_2016}              & 80.55 & 18.31 \\
Margolus~\cite{maslov_advantages_2016}                      & 26.62 & 13.72 \\
\midrule
$\mathrm{CNOT}_{4,1}$                                       & 31.99 & 28.46 \\
Double Peres~\cite{peres_reversible_1985}                   & 12.96 & 18.67 \\
Quantum Full Adder~\cite{hung_optimal_2006}                 & 203.64  & 153.6 \\
Double Toffoli~\cite{hung_optimal_2006}                     & 150.87  & 117.64 \\
\midrule
Average & 48.6& 38.7 \\
\bottomrule
\end{tabular}
\end{table}

\footnotetext{For the \emph{Magic} instance, due to the large elementary-gate set, we enabled redundancy–elimination only up to pairs; enumerating all redundant triplets would already yield $\sim\!1.3\times 10^{5}$ combinations. We recommend this option whenever the elementary gate set is very large (e.g., when using a fine–grained discretization of general single–qubit rotations).}

\subsection{Speed-ups from Valid Constraints}\label{subsec:speedups}
Having identified the target constraint to use going forward, we now evaluate the impact of the valid speed-up constraints. We begin with those tied directly to the target relation, namely the Hermitian-conjugate constraints: a linear HC--1 form and a quadratic MIQP variant, HC--2, as presented in Section \ref{subsec:hermitian}. Table~\ref{tab:mat-constraints} shows that these constraints can substantially reduce runtime—often by factors of two to five—though the magnitude depends strongly on the target.

Each run (including the \emph{base} column) uses total gate-count optimization with only the necessary baseline constraints (conditions \eqref{eq:gate_select} to \eqref{eq:G-hat}) plus the identity-gate symmetry constraint~\eqref{eq:IGS-constraint}, which we regard as fundamental. The subsequent columns add only the named constraint family to that baseline. The \emph{best MILP} column enables all three speed-up families simultaneously, namely redundancies, equivalent pairs, and HC 1 as well.

\begin{table*}[t]
\centering
\caption{Efficacy of valid constraints on wall-clock runtime (seconds) across circuit decompositions. Speedup is reported relative to the base run (values $>1$ indicate speedup). “best MILP” = base $+$ HC--1 $+$ redundancy $+$ equivalent-pair constraints.}
\label{tab:mat-constraints}
\footnotesize
\setlength{\tabcolsep}{4pt}
\renewcommand{\arraystretch}{1.12}
\resizebox{\textwidth}{!}{%
\begin{tabular}{lrrrrrrrrrrr}
\toprule
& \multicolumn{1}{c}{base} & \multicolumn{2}{c}{base + HC--1} & \multicolumn{2}{c}{base + HC--1 + HC--2} & \multicolumn{2}{c}{base + redundancies} & \multicolumn{2}{c}{base + equivalent pairs} & \multicolumn{2}{c}{best MILP} \\
\cmidrule(lr){2-2}\cmidrule(lr){3-4}\cmidrule(lr){5-6}\cmidrule(lr){7-8}\cmidrule(lr){9-10}\cmidrule(lr){11-12}
Target & time (s) & time (s) & speed-up & time (s) & speed-up & time (s) & speed-up & time (s) & speed-up & time (s) & speed-up \\
\midrule
Toffoli (with 2-qubit gates)~\cite{hung_optimal_2006} & 10.78   & 12.36   & 0.9x &  5.64   & 1.9x & 10.22   & 1.1x &  4.10   & 2.6x &  2.81   &  3.8x \\
$\mathrm{CNOT}_{1,3}$      & 101.39  & 52.53   & 1.9x & 31.90   & 3.2x & 25.03   & 4.1x &  9.19   &11.0x &  2.34   & 43.3x \\
Fredkin                    & 381.45  & 389.53  & 1.0x & 212.70  & 1.8x & 259.21  & 1.5x & 160.87  & 2.4x & 40.13   &  9.5x \\
Miller                     & 138.31  & 138.88  & 1.0x & 108.98  & 1.3x & 134.30  & 1.0x & 180.10  & 0.8x & 117.70  &  1.2x \\
Relative Toffoli           & 182.30  & 94.99   & 1.9x & 67.15   & 2.7x & 76.79   & 2.4x & 68.71   & 2.7x & 18.31   & 10.0x \\
Margolus                   & 294.27  & 87.85   & 3.3x & 124.73  & 2.4x & 243.73  & 1.2x & 37.42   & 7.9x & 13.72   & 21.4x \\
Quantum Fourier Transform  & 481.91  & 279.58  & 1.7x & 89.62   & 5.4x & 202.19  & 2.4x & 321.47  & 1.5x & 72.50   &  6.6x \\
Controlled-iSwap~\cite{rasmussen_simple_2020}         & 52090.81& 16008.00& 3.3x & 14501.24& 3.6x & 7447.50 & 7.0x & 10399.80& 5.0x & 1367.02 & 38.1x \\
\midrule
Double Peres               & 116.73  & 91.58   & 1.3x & 62.87   & 1.9x & 121.49  & 1.0x & 44.50   & 2.6x & 18.67   &  6.3x \\
Quantum Full Adder         & 2348.01 & 348.47  & 6.7x & 231.44  &10.1x & 1054.85 & 2.2x & 231.06  &10.2x & 153.60  & 15.3x \\
Double Toffoli             & 1583.08 & 1020.68 & 1.6x & 315.30  & 5.0x & 986.62  & 1.6x & 179.38  & 8.8x & 117.64  & 13.5x \\
\bottomrule
\end{tabular}%
}
\end{table*}

A few takeaways are worth highlighting. First, the pruning conditions (redundancies and equivalent pairs) often help the solver dramatically—up to $11\times$ on their own (see $\mathrm{CNOT}_{1,3}$ with equivalent pairs). Second, in 5 of the 11 targets, enabling all three speed-up families produces a combined gain that is larger than what one might expect from simply compounding their individual effects. For example, for the Fredkin gate, a naive expectation from individual columns suggests only about a $3.4\times$ improvement, yet the full combination delivers $9.5\times$. Similarly, for the Double Peres gate, the individual contributions amount to about $3.2\times$, while the combined run achieves $6.3\times$. Finally, across our test set, these three families together yield a maximum observed improvement of $43\times$ (for $\mathrm{CNOT}_{1,3}$). 

These speedups have a practical consequence — as seen in the Controlled-iSwap case — because without these constraints, the raw MILP runtimes make our approach essentially infeasible, whereas with them, the problem becomes tractable, broadening the set of target gates for which the package could be genuinely useful.

\subsection{Circuit Depth Minimization}\label{subsec:circuit-depth-opt}
Having evaluated the target constraints and the benefits of the three families of speed-ups, we now illustrate with a few numerical examples how \texttt{QCOpt} tackles \emph{depth optimization} directly, using a gate set similar to the one employed above. All runs use our best available constraints: the base formulation \eqref{eq:gate_select}–\eqref{eq:depth_mip} together with the exact-target constraint \eqref{eq:exact_target}, plus all speed-up constraints—redundancy constraints and the Hermitian-conjugate constraints (HC1, HC2). From the symmetry family, we include only the identity-gate symmetry \eqref{eq:IGS-constraint}; the equivalent-pair symmetry in \eqref{eq:commuting_pair} is not a valid constraint for depth optimization, as noted in Remark \ref{rem:depth_no_commute_cons}.

\begin{table}[]
   \centering
   \caption{Solver times for the task of circuit depth optimization.}
   \label{tab:direct-depth}
   \begin{tabular}{lr} 
      \toprule
      \textbf{Target} & \textbf{Solve time [s]} \\
      \midrule
      Toffoli (with 2-qubit gates) & 1.48 \\
      $\mathrm{CNOT}_{1,3}$ & 770.66 \\
      Fredkin & 82.60 \\
      Miller & 114.50 \\
      Relative Toffoli & 123.13 \\
      Margolus & 208.10 \\
      Quantum Fourier Transform & 893.79 \\
      \midrule
      Double Peres & 23.97 \\
      Quantum Full Adder & 232.60 \\
      Double Toffoli & 375.91 \\
      \bottomrule
   \end{tabular}
\end{table}

We observe that running the direct depth-minimization MIP typically requires substantially more time than gate-count optimization with this formalism, which matches expectations: depth minimization is both a harder objective and it introduces more decision variables and linking constraints than the simpler gate-budget formulation.

\subsection{Approximate Compilation with Fibonacci Anyons}\label{subsec:approx-fib-new}
We tested our approximate formalism on a topological model — computation with Fibonacci anyons — where two closely related gate pictures exist: braiding (the braid generators $\sigma_1,\sigma_2$) and weaving (their powers, typically $\sigma_1^2,\sigma_2^2$). For single-qubit compilation, one uses three anyons, yielding a two-dimensional fusion space acted on by two braiding operators. Let \(\varphi \coloneqq \tfrac{1+\sqrt{5}}{2}\) be the golden ratio.
Within an \(\mathcal{SU}(2)\) representation, the braiding operators can be written compactly as
\[
\begin{aligned}
\sigma_1 &= e^{i\pi/10}\, R, \qquad \sigma_2 = e^{i\pi/10}\, F R F, \quad \text{where} \\
R &= 
\begin{pmatrix}
e^{-4 i\pi/5} & 0\\
0 & e^{3 i\pi/5}
\end{pmatrix}, \quad
F = 
\begin{pmatrix}
\varphi^{-1} & \varphi^{-1/2}\\
\varphi^{-1/2} & -\varphi^{-1}
\end{pmatrix}.
\end{aligned}
\]

The weaving operators are $W_1=\sigma_1^{2}$ and $W_2=\sigma_2^{2}$, and as the braiding ones, they do not commute ($[W_1,W_2]\neq0$). Sequences drawn from $\{\sigma_1^{\pm1},\sigma_2^{\pm1}\}$ (or equivalently from their weaves and complex conjugates) are dense in $\mathcal{SU}(2)$, so any single-qubit unitary can be approximated to arbitrary precision. Prior work has largely searched for the best short sequences by brute force: for instance, a depth-15 weave achieving a Hadamard has reported fidelity $0.999957$~\cite{rouabah_compiling_2021}. To mirror approximate compilation methods in Section~\ref{subsec:approx}, we benchmark three objectives—ordered by increasing modeling tractability—on the depth-15 instances for the Hadamard, X, and $T$ gates; results appear in Table~\ref{tab:approx-fib-3targets}.

\begin{table*}[t]
\centering
\caption{Approximate-compilation results for depth-15 Fibonacci weaves across three targets: the number of outer approximators is fixed \(K=5\), varying the entrywise error bound \(\epsilon\) on the complex error matrix \(E\) (i.e., \(|E_{ij}|\leqslant \epsilon\)).}
\label{tab:approx-fib-3targets}
\footnotesize
\setlength{\tabcolsep}{6pt}
\renewcommand{\arraystretch}{1.12}
\begin{tabular}{r r r r r r r}
\toprule
\multicolumn{1}{c}{\textbf{Method}} &
\multicolumn{2}{c}{\textbf{$\bm{H}$-gate}} &
\multicolumn{2}{c}{\textbf{$\bm{X}$-gate}} &
\multicolumn{2}{c}{\textbf{$\bm{T}$-gate}} \\
\cmidrule(lr){2-3}\cmidrule(lr){4-5}\cmidrule(lr){6-7}
 & \multicolumn{1}{c}{\textbf{time (s)}} & \multicolumn{1}{c}{\textbf{fidelity}} &
 \multicolumn{1}{c}{\textbf{time (s)}} & \multicolumn{1}{c}{\textbf{fidelity}} &
 \multicolumn{1}{c}{\textbf{time (s)}} & \multicolumn{1}{c}{\textbf{fidelity}} \\
\midrule
outer-approximation ($\epsilon=1$):      & 1.06  & 0.945974 & 1.38  & 0.932852 & 1.01   & 0.973666 \\
outer-approximation ($\epsilon=1/2$):  & 1.83  & 0.987268 & 1.68  & 0.985019 & 1.49   & 0.979530 \\
outer-approximation ($\epsilon=1/4$):  & 6.86  & 0.992066 & 4.49  & 0.995352 & 9.67   & 0.998256 \\
outer-approximation ($\epsilon=1/8$):  & 8.25  & \textbf{0.999957} & 20.42 & \textbf{0.999990} & 13.29  & 0.998153 \\
outer-approximation ($\epsilon=1/16$): & 98.46 & 0.999921 & 56.12 & 0.999990 & 846.51 & \textbf{0.999739} \\
\midrule
linearized fidelity (MILP):                      & 217.66 & 0.999921 & 76.99  & 0.999990 & 30293  & 0.999739 \\
exact fidelity (MIQP):                    & 198.41 & 0.999957 & 191.06 & 0.999990 & 6897.8 & 0.999917 \\
\bottomrule
\end{tabular}
\end{table*}

First, we solve the exact, phase-invariant fidelity MIQP \eqref{eq:exact-fid-obj}, maximizing ($F(\cdot)=\alpha^{2}+\beta^{2}$). This non-convex objective is computationally heavier but yields certificates of optimality. On the depth-15 Fibonacci-anyon benchmarks, the MIQP recovers the best-known Hadamard and (X)-gate weaves within 200 s, substantially faster than exhaustive search (the only alternative that, to our knowledge, guarantees optimality but scales exponentially). Although mixed-integer optimization is NP-hard in the worst case, these results show that a carefully tailored model with valid constraints can solve practically relevant instances; the (T)-gate remains harder ($\approx$ 6900 s), motivating surrogate objectives. The MIQP’s certified fidelities serve as upper bounds for the surrogate objectives evaluated below.

Second, we evaluate the linearized fidelity surrogate \eqref{eq:lin-fid-obj-again}, which maximizes the real trace overlap ($\alpha$) (a linear objective). When the phase quadrature ($\beta\approx 0$), maximizing ($\alpha$) is equivalent to maximizing fidelity, and indeed it matches the MIQP on the (X)-gate and is near-optimal on the others. However, its runtimes are not uniformly better—and degrade notably for the ($T$) target (Table~\ref{tab:approx-fib-3targets})—motivating a third, more scalable objective.

Third, we evaluate the Frobenius‐error objective \eqref{eq:oa-frob}, realized via the linear outer approximation \eqref{eq:oa-tangents} using $K$ tangents per error term on a uniform grid in $[-\epsilon,\epsilon]$. This yields an MILP. By Prop.~\ref{prop:objective-equivalence}, minimizing $\|E\|_F^2$ is exactly equivalent to maximizing $\alpha$, while offering a significant run time advantage. With $K=5$ tangents per entry, decreasing $\epsilon$ tightens the outer approximation; runtimes increase monotonically while fidelities typically improve (Table~\ref{tab:approx-fib-3targets}). At $\epsilon=\tfrac{1}{8}$, the surrogate reaches $F=0.999957$ for the Hadamard in $8.25$~s and $F=0.999990$ for the X-gate in $20.42$~s, essentially matching the MIQP certified optima at a fraction of the time. For the $T$-gate, the surrogate yields high-quality but suboptimal fidelity—e.g., $F=0.999739$ at $\epsilon=\tfrac{1}{16}$ versus $0.999917$ from the MIQP—reflecting that controlling $\|E\|_F^2$ (equivalently, $\alpha$) cannot always suppress a residual phase quadrature $\beta$.

For reference, the exact-fidelity MIQP produced the following certified depth-15 weaves:
\begin{equation*}
\resizebox{\columnwidth}{!}{$
\begin{aligned}
H &\simeq \sigma_{1}^{-4}\,\sigma_{2}^{+2}\,\sigma_{1}^{-2}\,\sigma_{2}^{+2}\,\sigma_{1}^{-2}\,\sigma_{2}^{-2}\,\sigma_{1}^{+2}\,\sigma_{2}^{-4}\,\sigma_{1}^{-2}\,\sigma_{2}^{+2}\,\sigma_{1}^{+2}\,\sigma_{2}^{-2}\,\sigma_{1}^{-2},\\
X &\simeq \sigma_{2}^{+2}\,\sigma_{1}^{-4}\,\sigma_{2}^{+2}\,\sigma_{1}^{-4}\,\sigma_{2}^{+2}\,\sigma_{1}^{-4}\,\sigma_{2}^{+2},\\
T &\simeq \sigma_{1}^{+2}\,\sigma_{2}^{-2}\,\sigma_{1}^{+2}\,\sigma_{2}^{+4}\,\sigma_{1}^{-2}\,\sigma_{2}^{+2}\,\sigma_{1}^{-2}\,\sigma_{2}^{-4}\,\sigma_{1}^{+2}\,\sigma_{2}^{-2}\,\sigma_{1}^{-4}\,\sigma_{2}^{-2}.
\end{aligned}
$}
\end{equation*}

\subsection{RHO with 3-body Interactions}\label{subsec:RHOexp1}

We test the rolling–horizon optimizer on circuits built from three–body $ZZZ$ interactions placed on hyperedges of a graph. Such $3$-local Ising terms are natural in spin models with multi–spin couplings~\cite{baxter_exact_1973}, and also appear as parity checks in subsystem/topological quantum error–correction architectures that use three–qubit $XXX/ZZZ$ measurements \cite{bravyi_subsystem_2013}.

Throughout, we take the phase-rotation convention
\begin{equation}
R_{ZZZ}(\theta)\;=\;\exp\!\big(-i\tfrac{\theta}{2}\,Z\otimes Z\otimes Z\big),
\end{equation}
so that a single $R_{ZZZ}(\theta)$ may be compiled as a short parity ladder: compute the $Z$–parity of two controls into an accumulator with two CNOTs, apply a single–qubit $R_Z(\theta)$ on the accumulator, and uncompute. In a circuit diagram, therefore, one interaction term is shown in Fig.~\ref{fig:ZZZ}.

\begin{figure}[]
  \centering
  \begin{center}
\begin{quantikz}
\lstick{$q_1$} & \ctrl{2} & \qw         & \qw                & \qw & \ctrl{2} & \qw \\
\lstick{$q_2$} & \qw      & \ctrl{1}    & \qw                & \ctrl{1} & \qw & \qw \\
\lstick{$q_3$} & \targ{}  & \targ{}     & \gate{R_Z(\theta)} & \targ{}  & \targ{} & \qw
\end{quantikz}
\end{center}
 \caption{Parity-ladder compilation of a three-body interaction $R_{ZZZ}(\theta)$ in a quantum circuit.}
  \label{fig:ZZZ}
\end{figure}
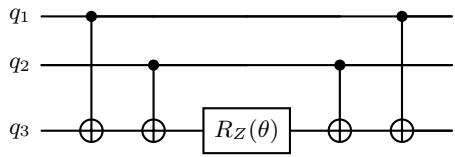

We consider the full hypergraph on 5 qubits, $K_5^{(3)}$, where there are $\binom{5}{3}=10$ hyperedges, and apply the maximally entangling phases $R_{ZZZ}(\pi/2)$ in lexicographic order to each 3-subset $\{i,j,k\}\subset\{1,\dots,5\}$ (with $1<2<3<4<5$ ordering). As a naïve target seed, we simply concatenated the decomposition above for each hyperedge (ten copies of a 5-gate pattern), and we restricted the elementary gate set to $\{\mathrm{CNOT},H,S\}$ with $S=R_Z(\pi/2)$. Running RHO with parameters \(\texttt{window\_length}=10\), \(\texttt{accept\_window}=5\), \(\texttt{max\_qubits}=4\), the optimizer produced a 36-gate circuit, saving 14 gates by sharing intermediate parities across successive terms. Although this is a remarkable advantage, this sequence is still suboptimal. A second RHO pass seeded by the 36-gate solution further reduced the count to 32 gates, indicating that while exact MILP becomes intractable at these sizes, the RHO heuristic can still realize substantial improvements on larger instances.

\subsection{RHO's Performance on a Larger Circuit}\label{subsec:RHOexp2}

We illustrate the performance of RHO—and, in particular, its dependence on the window parameters—on a four-qubit benchmark from Ref.~\cite{revbench_4b15g1}. The circuit (Fig.~\ref{fig:rho-benchmark}) comprises 15 gates from $\{\mathrm{X}, \mathrm{CNOT}=\mathrm{CX}, \mathrm{CCX}, \mathrm{CCCX}\}$; our objective is to obtain an implementation over the canonical Clifford+T set $\{H,T,\mathrm{CNOT}\}$.

To make the instance compatible with the RHO wrapper, we first translate each gate to the canonical set. We represent an $X$ as $H T^4 H$ (two Hadamards and four $T$ gates), keep $\mathrm{CNOT}$ as is, and use a standard $T$-optimal Toffoli with $6$ CNOTs, $7$ $T$-type gates (counting $T$/$T^\dagger$), and $2$ Hadamards. For the four-controlled NOT ($\mathrm{CCCX}$), we follow a Barenco-style decomposition~\cite{barenco_elementary_1995} into $6$ CNOTs and $7$ $\mathrm{controlled-\sqrt{X}}$ gates; each $\mathrm{controlled-\sqrt{X}}$ is then expressed using $2$ Hadamards, $2$ CNOTs, and $3$ $T$-type gates. After local $H$-cancellations this yields a 43-gate realisation of $\mathrm{CCCX}$ over $\{H,T,\mathrm{CNOT}\}$. Overall, the benchmark translates to a seed of $142$ one- and two-qubit gates on which we run the RHO wrapper. As a heuristic baseline, transpiling the same 142-gate seed with Qiskit \texttt{2.2.0}~\cite{aleksandrowicz_qiskit_2019} reduces it to 132 gates.

\begin{figure}[!h]
  \centering
  \begin{center}
\begin{quantikz}[row sep=0.4cm, column sep=0.25cm]
\lstick{$q_1$} &
\ctrl{2} & \qw     & \targ{} & \qw     & \ctrl{1} & \qw     & \ctrl{3} & \targ{} & \qw          & \ctrl{3} & \qw & \targ{} & \ctrl{1} & \qw & \qw \\
\lstick{$q_2$} &
\qw      & \qw     & \qw     & \ctrl{1}& \targ{}  & \targ{} & \ctrl{2} & \qw     & \targ{}      & \qw & \ctrl{1}& \ctrl{-1}& \targ{}  & \qw & \qw\\
\lstick{$q_3$} &
\targ{}  & \ctrl{1}& \qw     & \targ{} & \qw      & \ctrl{-1}& \ctrl{1} & \ctrl{-2}    & \targ{} & \qw      & \targ{}  & \ctrl{-2} & \ctrl{-1}& \targ{} & \qw\\
\lstick{$q_4$} &
\qw      & \targ{} & \ctrl{-3}& \ctrl{-1}& \qw     & \ctrl{-2}& \targ{} & \qw     & \qw          & \targ{}  & \ctrl{-1}& \qw      & \qw  & \qw & \qw
\end{quantikz}
\end{center}

  \caption{Four-qubit benchmark circuit used to probe RHO parameter sensitivity.}
  \label{fig:rho-benchmark}
\end{figure}
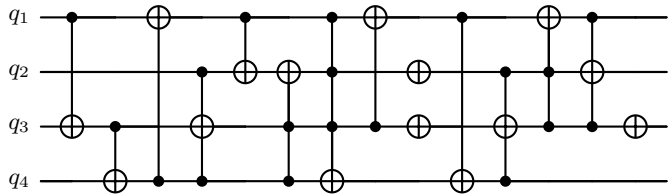

We then explore three representative settings for the RHO parameters \texttt{window\_length} ($L_W$), \texttt{accept\_window} ($A_W$), and \texttt{max\_qubits} ($Q_W$):

\begin{itemize}
  \item \emph{Large blocks with moderate qubit budget:} $(L_W,A_W,Q_W)=(10,5,3)$. This ambitious choice ran for $>24$\,h in total, with several windows timing out after $>4$\,h (solutions found but without optimality certificates). The pass reduced the seed from 142 to 122 gates. Given the exaggerated runtime, we deemed iterative passes infeasible in this configuration.

  \item \emph{Very small qubit budget:} $(10,5,2)$. Here, most windows are too small to expose non-local simplifications, but the runtime is $<2$\,min. The first pass reduced 142$\to$130 gates, a second pass reached 126, and a third pass produced no further improvement.

  \item \emph{Balanced choice:} $(7,5,3)$. This setting provided the best trade-off. The first pass reduced 142$\to$126, a second pass reached 116, and subsequent passes did not improve further. Each pass completed within 90\,min.
\end{itemize}

Overall, we find that (i) overly aggressive windows with insufficient qubit caps incur prohibitive solve times per window; (ii) small $Q_{W}$ makes passes fast but myopic; and (iii) a balanced choice (here, $L_W=7$, $A_W=5$, $Q_W=3$) yields substantial reductions with manageable runtime, enabling effective multi-pass refinement. For further context, as Qiskit produces 132 gates on the same instance, whereas our best RHO setting reaches 116 (16 fewer than Qiskit, $\approx 12.1\%$ reduction; 26 fewer than the seed, $\approx 18.3\%$).

These experiments illustrate that RHO can exploit structure that is invisible to a purely local, per-term compilation, generating shorter circuits even when exact global optimization methods on the entire circuit are computationally intractable.

\section{Conclusions}

In conclusion, we present a depth-aware mixed-integer programming framework for quantum circuit compilation that treats global phase with linear constraints, optimizes depth via explicit scheduling, and sharply tightens the search space with domain-specific valid inequalities. Empirically, enforcing the exact $\mathcal{SU}(2^Q)$ target and enabling our cuts reduces wall-clock time by large factors—often an order of magnitude and up to $\sim 40\times$—making exact certification practical on nontrivial medium-scale circuits. For approximate synthesis, linear surrogates expose accuracy–time trade-offs, while a non-convex, phase-invariant fidelity objective attains certified, best-known solutions (e.g., Hadamard with depth-15 Fibonacci-anyon weaves with $F\approx 0.999957$), and faster linearized objectives deliver the same fidelity faster. Beyond exact MIP’s reach, a rolling-horizon strategy that rolls primarily in time and enforces per-qubit closure preserves local context and yields iterative gains, including parity sharing. Empirically, RHO provides substantial gate savings on multi-body parity circuits (e.g., $50 \to 32$ gates over two passes). On a 142-gate seed circuit, RHO produces 116 gates (an 18.3\% reduction, 12.1\% fewer than a Qiskit baseline).

Methodologically, the results show that (i) linear handling of global phase and explicit depth scheduling make exact depth optimization practical at \textit{modest scales}; (ii) well-chosen, domain-aware constraints can shift instances from unsolvable to tractable; and (iii) fidelity-driven approximate formulations admit both exact (non-convex) and convex/MILP surrogates with provable relationships. Practically, the same optimization stack accommodates hardware-aware costs, topology constraints, and routing-aware compilations, offering a single point of control over objectives that matter in both the NISQ and fault-tolerant regimes.

Overall, the approach unifies provable optimality and principled heuristics within the open-source \texttt{QCOpt} package, offering a practical path to hardware-aware compilation; promising extensions include richer cut libraries, adaptive horizons and warm starts, and tighter integration with state-of-the-art hardware-compatible gates.

\begin{acknowledgments}
The authors gratefully acknowledge support from the U.S. Department of Energy (DOE) Quantum Computing Program, sponsored by the Los Alamos National Laboratory (LANL) Information Science \& Technology Institute, as well as funding from LANL’s Laboratory Directed Research and Development (LDRD) program under project ``20230091ER: Learning to Accelerate Global Solutions for Non-convex Optimization''. Z.S. also gratefully acknowledges stimulating discussions with Gavin Brennen and Max Hunter Gordon, and support from the Sydney Quantum Academy.
\end{acknowledgments}

\appendix

\section{Auxiliary Parameters for Example Targets}\label{app:example-params}

Table~\ref{tab:aux-target-params} from this appendix collects the parameters used in Sec.~\ref{sec:examples}. 
Here $Q$ is the number of qubits for the target, $|\mathbb{G}|$ is the
cardinality of the elementary gate set made available to the optimizer, and
$P$ is the maximum gate count allowed for the synthesized circuit.

\begin{table}[]
  \centering
  \footnotesize
  \caption{Parameters used for the targets in Sec.~\ref{sec:examples}.}
  \label{tab:aux-target-params}
  \begin{tabular}{@{}lccc@{}}
    \toprule
    \textbf{target} & $Q$ & $|\mathbb{G}|$ & $P$ \\
    \midrule
    \multicolumn{4}{c}{\emph{Two-qubit targets}}\\
    \midrule
    Controlled-$\sqrt{X}$                & 2 & 9   & 7  \\
    Controlled-Hadamard                   & 2 & 32  & 5  \\
    Magic                                 & 2 & 73 & 4  \\
    iSwap                      & 2 & 9   & 10 \\
    single-excitation Hadamard            & 2 & 14  & 5  \\
    \midrule
    \multicolumn{4}{c}{\emph{Three-qubit targets}}\\
    \midrule
    Toffoli (with 2-qubit gates)            & 3 & 9   & 5  \\
    $\mathrm{CNOT}_{1,3}$                 & 3 & 14  & 8  \\
    Fredkin                               & 3 & 11  & 7  \\
    Miller                                & 3 & 7   & 10 \\
    Relative Toffoli                      & 3 & 7   & 9  \\
    Margolus                              & 3 & 12  & 7  \\
    Quantum Fourier Transform             & 3 & 17  & 7  \\
    Controlled-iSwap           & 3 & 7   & 12 \\
    \midrule
    \multicolumn{4}{c}{\emph{Four-qubit targets}}\\
    \midrule
    $\mathrm{CNOT}_{4,1}$                 & 4 & 6   & 10 \\
    Double Peres                          & 4 & 9   & 7  \\
    Quantum Full Adder                    & 4 & 11  & 7  \\
    Double Toffoli                        & 4 & 10  & 7  \\
    \bottomrule
  \end{tabular}
\end{table}

\bibliography{references}

\end{document}